%% file: main.tex
\documentclass{lmcs}
\pdfoutput=1
\usepackage[utf8]{inputenc}

\usepackage{lastpage}
\lmcsdoi{21}{1}{10}
\lmcsheading{}{\pageref{LastPage}}{}{}%
{Mar.~07,~2022}{Jan.~29,~2025}{}

\usepackage{amsmath,amsfonts,amsthm,amssymb,multirow}

\usepackage{tikz}

\usepackage{graphicx}
\usepackage{color}
\usepackage[new]{old-arrows}
\usepackage{xspace}
\usepackage[mathscr]{euscript}
\usepackage{url}
\usepackage{latexsym,epsfig}
\usepackage[normalem]{ulem}
\usepackage{bm}

\input{macros}

\begin{document}

\title[Multi-Structural Games and Number of Quantifiers]{Multi-Structural Games and Number of Quantifiers$^*$}
\titlecomment{$^*$This paper is an expanded and reorganized version of the LICS (Logic in Computer Science) 2021 paper \cite{Fagin21}, which points to the Computer Science arXiv, version 1, entry \cite{Fagin21a} for full proofs of its theorems.}

\author[R.~Fagin]{Ronald Fagin\lmcsorcid{0000-0002-7374-0347}}[a]
 \address{IBM Research - Almaden, San Jose, CA}
 \email{fagin@us.ibm.com}

\author[J.~Lenchner]{Jonathan Lenchner\lmcsorcid{0000-0002-9427-8470}}[b]
 \address{IBM T.J. Watson Research Center, Yorktown Heights, NY}
 \email{lenchner@us.ibm.com}

\author[K.~W.~Regan]{Kenneth W. Regan\lmcsorcid{0000-0002-7382-7178}}[c]
 \address{Department of CSE, University at Buffalo, Amherst, NY}
 \email{regan@buffalo.edu}

\author[N.~Vyas]{Nikhil Vyas\lmcsorcid{0000-0002-4055-7693}}[d, \begin{minipage}{1em}\tiny$\dagger$\end{minipage}]
 \address{Department of EECS, MIT, Cambridge, MA}
 \email{nikhil@g.harvard.edu} 
 \thanks{$^\dagger$Supported by NSF CCF-1909429. Work Partially done while visiting IBM Research - Almaden}

\begin{abstract}
We study multi-structural games, played on two sets $\mathcal{A}$ and $\mathcal{B}$ of structures. These games generalize \ef games. Whereas \ef games capture the quantifier rank of a first-order sentence, multi-structural games capture the number of quantifiers, in the sense that Spoiler wins the $r$-round game if and only if there is a first-order sentence $\phi$ with at most $r$ quantifiers, where every structure in $\mathcal{A}$ satisfies $\phi$ and no structure in $\mathcal{B}$ satisfies $\phi$. We use these games to give a complete characterization of the number of quantifiers required to distinguish linear orders of different sizes and we develop machinery for analyzing structures beyond linear orders.
\end{abstract}

\maketitle

\section{Introduction}
Model theory has a number of techniques for proving inexpressibility results. 
However, 
as noted in \cite{Fagin93},
almost none of the key theorems and tools of model
theory, such as the compactness theorem and the
L\"owenheim-Skolem theorems, apply to
finite structures.  Among the few tools of model
theory that yield inexpressibility results for finite structures are \ef games \cite{Ehr61,Fra54}, henceforth \edashf games. In this paper we shall entirely be concerned with \textit{finite} structures over \textit{finite} relational vocabularies. For definitions of these concepts, we refer the reader to \cite{Imm99}. 

The standard \edashf game is played by ``Spoiler'' and ``Duplicator'' on a pair $(A,B)$ of structures over the same first-order vocabulary $\tau$, for a specified number $r$ of \emph{rounds}.  In each round, Spoiler chooses an element from $A$ or from $B$, and Duplicator replies by choosing an element from the other structure.  In this way, they determine sequences of elements $a_1,\dots,a_r \in A$ and $b_1,\dots,b_r \in B$, repetitions allowed, which, in addition to any possible constants defined on $A$ and $B$, define substructures $A'$ of $A$ and $B'$ of $B$.  The analysis can be phrased as attempting to define a sequence of functions $f$, $g$, and $f\cup g$. If any of these functions are not well-defined then Spoiler wins. The function $f{:}\{a_1,....,a_r\} \rightarrow \{b_1,....,b_r\}$ is defined by $f(a_i) = b_i$ for $i = 1,...,r$. If there are $i$ and $j$ with $1 \leq i < j \leq r$ such that $a_i = a_j$ but $b_i \neq b_j$, then $f$ is not well-defined. The partial function $g{:} A \rightarrow B$ is defined so that $g(a) = b$ if $a$ and $b$ are associated with the same constant symbol in $\tau$. If after a round of play there are two constant symbols $c_i, c_j$ with $a \in A$ associated with both $c_i$ and $c_j$, but different elements of $B$ associated with $c_i$ and $c_j$, then $g$ is not well-defined.  The joint partial function $f \cup g$ stipulates that $(f \cup g)(a) = f(a)$ if $f$ is defined on $a$ and $(f \cup g)(a) = g(a)$ if $g$ is defined on $a$.  If $f$ and $g$ \emph{are} well-defined \emph{and} agree on common elements, then $f \cup g$ is well-defined.  If it is not well-defined, or if $f \cup g$ is not an isomorphism from its domain $A'$ to its image $B'$, then Spoiler wins. Finally, if  $f \cup g$ is well-defined, and $f \cup g$ is an isomorphism from $A'$ to $B'$, then Duplicator wins. In this latter case of a Duplicator win, we say that the sequences $\langle a_1,...,a_r \rangle$ and  $\langle b_1,...,b_r \rangle$ of selected elements ``give rise to a partial isomorphism'' between $A$ and $B$.

The equivalence theorem for \edashf games \cite{Ehr61,Fra54} 
characterizes the minimum \emph{quantifier rank} of a sentence $\phi$ over $\tau$ that is true for  $A$ but false for $B$.  
The quantifier rank $\qr(\phi)$ is defined as zero for a quantifier-free sentence $\phi$, and inductively:
\begin{eqnarray*}
\qr(\neg\phi) &=& \qr(\phi),\\
\qr(\phi \vee \psi) =  \qr(\phi \land \psi)&=& \max\set{\qr(\phi),\qr(\psi)},\\
\qr(\forall x \phi(x)) = \qr(\exists x \phi(x))  &=& \qr(\phi) + 1,
\end{eqnarray*}

\begin{thm}[\textbf{Equivalence Theorem for E-F Games}]\label{thm:ef}
Spoiler wins the $r$-round \edashf game on $(A,B)$ if and only if there is a 1st order sentence $\phi$ of quantifier rank at most $r$ such that $A \models \phi$ while $B \models \neg\phi$.
\end{thm}

The ``if'' direction of this theorem is fairly easy to prove by induction on $r$.  This is the ``useful" direction,  which is used to prove inexpressibility results. The ``only if'' direction is somewhat tricky to prove;  intuitively, it tells us that 
any technique for proving
that a certain property cannot be defined by a first-order sentence with a certain quantifier rank 
can, in principle, be replaced by a proof via \edashf games. See \cite{Imm99, Lib12} for 
a proof and extended discussion.

We investigate a variant of \edashf games that we call \emph{multi-structural games}. These games  make Duplicator more powerful and characterize the \emph{number} of quantifiers rather than quantifier \emph{rank}.
It is straightforward to see that the minimum number of quantifiers needed to define a property $P$ is the same as the minimum size of the quantifier prefix of a sentence in prenex normal form that is needed to define property $P$.
This is because converting a sentence into prenex normal form does not increase the number of quantifiers.

As we discovered during review of this paper's conference version [FLRV21a] and acknowledged there, an equivalent of our multi-structural game was described in the journal version of Neil Immerman’s paper ``Number of Quantifiers Is Better Than Number of Tape Cells" \cite{Immerman81}. Its conference version \cite{Immerman79} did not mention the game.\footnote{We reference two personal communications with Immerman as \cite{Immerman21}, one beforehand where he raised the related size game described in \cite{AdlImm03} and discussed in section \ref{sec:related_work}, and one after our discovery that informed the present discussion.}  In \cite{Immerman81},  Immerman called it the ``separability game" and showed that it characterizes the number of quantifiers, without providing further results.
Just prior to the conclusion of \cite{Immerman81}, Immerman remarked,

\begin{small}
\begin{quote}
``Little is known about how to play the separability game. We leave it here as a 
jumping off point for further research. We urge others to study it, hoping that the 
separability game may become a viable tool for ascertaining some of the lower 
bounds which are `well believed' but have so far escaped proof.''
\end{quote}
\end{small}

\noindent Indeed, as our paper shows, analysis of the multi-structural games is often quite confounding, with many delicate issues. 

We now define the rules of the multi-structural game (henceforth, ``MS game''). There are again two players, Spoiler and Duplicator, and there is a fixed number $r$ of rounds.  Instead of being played on a pair 
$(A,B)$
of structures with the same vocabulary (as in an \edashf game), it is played on a pair $({\mathcal A}, {\mathcal B})$ of sets of structures, all with the same vocabulary.  
For $k$ with $0 \leq k \leq r$, by a   
\textit{labeled structure} after $k$ rounds, we mean a structure along with a labeling of which elements were selected from it in each of the first $k$ rounds. 
Let ${\mathcal{A}}_0 = \mathcal{A}$ and ${\mathcal{B}}_0 = \mathcal{B}$.
Thus, ${\mathcal{A}}_0$ represents the labeled structures from $\mathcal A$ after 0 rounds, and similarly for ${\mathcal{B}}_0$.
If $1 \leq k <r$, let ${\mathcal{A}}_k$ be the labeled structures originating from $\mathcal A$ after $k$ rounds, and similarly for ${\mathcal{B}}_k$.
In round $k+1$, 
Spoiler either chooses an element from each member of ${\mathcal{A}}_k$, thereby creating ${\mathcal{A}}_{k+1}$,
or chooses an element from each member of ${\mathcal{B}}_k$, thereby creating ${\mathcal{B}}_{k+1}$.
Duplicator responds as follows. 
Suppose that Spoiler chose an element from each member of ${\mathcal{A}}_k$, thereby creating ${\mathcal{A}}_{k+1}$.   Duplicator can then make multiple copies of each labeled structure of ${\mathcal{B}}_k$, and choose an element from each copy,
thereby creating ${\mathcal{B}}_{k+1}$.
Similarly, if Spoiler chose an element from each member of ${\mathcal{B}}_k$, thereby creating ${\mathcal{B}}_{k+1}$, Duplicator can then make multiple copies of each labeled structure of ${\mathcal{A}}_k$, and choose an element from each copy, thereby creating ${\mathcal{A}}_{k+1}$.
Duplicator wins if there is some labeled $A$ in ${\mathcal{A}}_{r}$ and some labeled $B$ in ${\mathcal{B}}_{r}$ where the labelings (in addition to any constants) give rise to a partial isomorphism in the same sense as in an E-F game.  Otherwise, Spoiler wins.

Note that on each of Duplicator's moves, Duplicator can make ``every possible choice," via the multiple copies. Making every possible choice creates what we call the \emph{oblivious strategy}.
It is easy to see that Duplicator has a winning strategy if and only if the oblivious strategy is a winning strategy.

We shall prove the following theorem.
It is analogous to Theorem~\ref{thm:ef} for ordinary \edashf games.

\begin{thm} \textbf{Equivalence Theorem for MS Games:} \label{thm:main1}
Spoiler wins the $r$-round MS game on
$(\mathcal{A}, \mathcal{B})$
if and only if there is a 1st order sentence $\phi$ with at most $r$ quantifiers such that $A \models \phi$ for every $A \in {\mathcal A}$ while $B \models \neg\phi$ for every $B \in {\mathcal B}$.
\end{thm}

We now give an interesting refinement of the Equivalence Theorem (although, as we shall discuss,  it does not seem to directly imply the Equivalence Theorem). 
Let $Q_1 \cdots Q_r$ be a sequence of quantifiers.  
We now define the ``$Q_1 \cdots Q_r$ MS game''.
It is an $r$-round MS game, with the following restrictions on Spoiler.
If $Q_k$ is an existential quantifier, then Spoiler's $k$th move 
must be in $\mathcal A$, and otherwise it must be in $\mathcal B$.
We then have the following result.

\begin{thm} \textbf{Fixed Prefix Equivalence Theorem for MS Games:} \label{thm:main1a}
Spoiler wins the $Q_1 \cdots Q_r$ MS game on
$(\mathcal{A}, \mathcal{B})$
if and only if there is a 1st order sentence $\phi$ 
in prenex normal form with exactly $r$ quantifiers, in the order $Q_1 \cdots Q_r$,
such that $A \models \phi$ for every $A \in {\mathcal A}$ while $B \models \neg\phi$ for every $B \in {\mathcal B}$.
\end{thm}
On the face of it, Theorem~\ref{thm:main1a} does not seem to directly imply Theorem~\ref{thm:main1}, for the following reason.
It is a priori conceivable that Spoiler's winning strategy in an $r$-round game is to move first in $\mathcal A$, and then, depending on Duplicator's response, to move in either $\mathcal A$ or $\mathcal B$.
So there is then no prefix $Q_1 \cdots Q_r$ dictating where Spoiler must move. 
In fact, in an \edashf game (but not in a MS game\footnote{Since we can assume Duplicator plays obliviously, Spoiler's moves are not conditioned on what Duplicator plays.}), it can indeed happen that Spoiler can win in 2 rounds, and where Spoiler  plays in the second round depends on how Duplicator played in the first round. 
An example of where this phenomenon can take place in a 2-round \edashf game is via the sentence 
$\exists x (\forall y B(x,y) \land \exists y R(x,y))$.
The proof of Theorem~\ref {thm:main1} appears in Section~\ref {sec:fundamental}.
The proof of Theorem~\ref{thm:main1a} 
 is almost the same as the proof of Theorem~\ref{thm:main1} and requires just very minor changes.

There is an interesting and non-obvious difference between \edashf games and MS games. Let us say that a player makes a move ``on top of'' a previous move if the player selects an element $c$ of a structure, and the same element $c$ had been selected by either player in an earlier round.
It is easy to see that in an \edashf game, it never helps Spoiler to make a move on top of a previous move (it only wastes a round). On the other hand, in MS games the issue of playing on top of a previous move is a frequent consideration for us. In fact, a detailed analysis shows that playing a move on top of a previous move may be a necessary part of a winning Spoiler strategy (see  Observation \ref{obs:play-on-top}).

We now give an example (Figure \ref{fig:3_vs_2_singleton-intro}) that shows differences between the \edashf game and the MS game, and what they say about quantifier rank vs. number of quantifiers. 
Consider the following two structures $B$ (for ``Big'') and $L$  (for ``Little''), over $\tau = \{<\}$, where $<$ is the binary ``less than'' relation. The vertex labels are not part of the structures.  Elements that appear to the left, within the same structure, are considered to be less than elements to the right. $B$ is a linear order on $3$ elements and $L$ is a linear order on two elements. In the text of this paper we write $B(i)$ (or $L(i)$) to denote the $i$th element in the linear order $B$ (respectively, $L$), while in the figures, for economy of space, we label the $i$th vertex instead by $Bi$ (respectively, $Li$).
\begin{figure} [ht]
\centerline{\scalebox{0.4}{\includegraphics{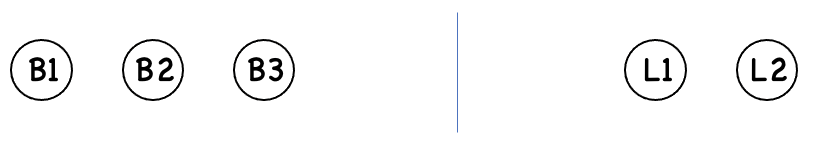}}}
\caption{An example showing the difference between MS and \edashf games.}
\label{fig:3_vs_2_singleton-intro}
\end{figure}
Further, rather than use the notation ${<}(x,y)$ 
we shall use the customary $x < y$ notation.

Suppose first that the number $r$ of rounds is $2$.
We show that Spoiler wins the \edashf game. 
On Spoiler's first move, Spoiler selects
 vertex $B(2)$ in $B$.  Duplicator must select either $L(1)$ or $L(2)$ in $L$. If Duplicator chooses $L(1)$, then Spoiler selects $B(1)$ in $B$.
 After Duplicator selects $L(2)$ in $L$ Spoiler wins since the mapping given by $B(2) \mapsto L(1)$ and $B(1) \mapsto L(2)$ fails to be a partial isomorphism because the ``less than'' relationship is flipped.  If Duplicator had instead selected $L(2)$ in the first round, then Spoiler would have won, by a similar argument, by selecting $B(3)$ in $B$ in the second round.
 
 The fact that Spoiler wins the 2-round game over $(B,L)$ tells us (by Theorem~\ref{thm:ef}) that there is a sentence
$\phi$ of quantifier rank at most 2 such that $B \models \phi$ while $L \models \neg\phi$.
Such a sentence is: $\exists x(\exists y(y < x) \AND \exists y(x < y))$.

Now let us consider the 2-round MS game over $({\mathcal B}, {\mathcal L})$, where ${\mathcal B} = \set{B}$ and ${\mathcal L} = \set{L}$.
We show 
that unlike the 2-round \edashf game, in the 2-round MS game, Duplicator wins. 
It is easy to see that Duplicator wins if Spoiler's first-round move is anything other than $B(2)$.
Let us see what happens if Spoiler's first move  is $B(2)$, which was a winning move for Spoiler in the 2-round E-F game.
Then on Duplicator's first move, Duplicator makes a second copy of $L$ and in one copy, call it $L_1$, Duplicator selects $L(1)$, and in the other copy, call it $L_2$, Duplicator selects $L(2)$.
Let us now consider Spoiler's possible second round responses. 
Suppose first that Spoiler's second round move is in $B$.
If Spoiler selects $B(3)$, then Duplicator selects $L(2)$ in $L_1$ and the mapping $B \rightarrow L_1$ such that $B(2)\mapsto L(1), B(3) \mapsto L(2)$ yields a partial isomorphism. On the other hand, if Spoiler selects $B(1)$, then Duplicator selects $L(1)$ in $L_2$ and $B \rightarrow L_2$ such that $B(1) \mapsto L(1), B(2) \mapsto L(2)$ yields a partial isomorphism. Section \ref{sec:upperbds} will complete the analysis of a Duplicator win.  
Since Duplicator wins the 2-round game, it follows 
by Theorem \ref{thm:main1} there is no sentence with just two quantifiers that distinguishes 
$\mathcal B$ from $\mathcal L$.

\medskip

The focus of our analysis of MS games in this paper is finite linear orders. Henceforth, all linear orders are assumed to be finite. In the case of E-F games, one has the following:

\begin{thmC}[\cite{Ros82}]\label{thm:f}
Let $f(r) = 2^r -1$.
In an $r$-round  \edashf  game played on two linear orders of different sizes, Duplicator wins if and only if the size of the smaller linear order is at least $f(r)$.
\end{thmC}

Since part of the proof of Theorem~\ref{thm:f} is left to an exercise in \cite{Ros82}, we give a proof 
in Appendix \ref{app:f}. 
Further,
the proof illustrates a simple recursive idea that is surprisingly not available to us in the analysis of linear orders from the vantage point of MS games. 

Theorem \ref{thm:f} together with Theorem \ref{thm:ef} imply that $f(r)$ is the maximum value $k$ such that a sentence of quantifier rank $r$ can distinguish linear orders of size $k$ and above from those of size smaller than $k$.

In analogy to this function $f$, and in an effort to arrive at a parallel theorem to Theorem \ref{thm:f}, we make the following definition.

\begin{defi} \label{def:g}
Define the function $g:\mathbb{N} \rightarrow \mathbb{N}$ such that $g(r)$ is the maximum number $k$ such that there is a sentence with $r$ quantifiers that can distinguish linear orders of size $k$ or larger, from linear orders of size less than $k$.
\end{defi}

To see that $g$ is well-defined, observe that the sentence 
\begin{equation} \label{eqn:basic}
	\exists x_1 \cdot\cdot\cdot \exists x_r \bigwedge\limits_{1 \leq i < r} x_i < x_{i+1},
\end{equation}
distinguishes linear orders of size $r$ or larger from linear orders of size less than $r$. Furthermore, 
there are only finitely many inequivalent sentences with up to $r$ quantifiers that include only the relation symbols $<$ and $=$, some fraction of which distinguish linear orders of some size $k$ or greater from linear orders of size less than $k$. There is therefore a maximum such $k \geq r$, which is then $g(r)$.

After building up quite a bit of machinery we eventually arrive at the following:

\begin{thm}\label{thm:g}
The function $g$ takes on the following values: 
$g(1) = 1, g(2) = 2, g(3) = 4, g(4) = 10$, and for $r > 4$,
\begin{equation*}
g(r) = \begin{cases} &2g(r-1)~~~~~~\textrm{if $r$ is even,}\\ &2g(r-1) + 1\textrm{ if $r$ is odd.} \end{cases}
\end{equation*}
\end{thm}

The value $g(3) = 4$ is a curious anomaly.  If it had turned out that $g(3) = 5$, then the entire induction could be founded on $r = 1$.
The proof of Theorem~\ref{thm:g} is a careful mathematical journey to restore this induction founded instead at $r = 4$. 

The following theorem for MS games is the analog of Theorem \ref{thm:f} for E-F games, and describes precisely when Duplicator (alternatively, Spoiler) wins $r$-round MS games on two linear orders of different sizes.

\begin{thm} \label{thm:g_for_game_play}
In an $r$-round MS game played on two linear orders of different sizes Duplicator has a winning strategy if and only if the size of the smaller linear order is at least $g(r)$. 
\end{thm}

It is important to note that neither the \textit{if} nor \textit{only if} portion of this theorem is implied by the definition of $g$.

\smallskip

Both \edashf games and MS games are used to prove inexpressibility results by showing there is a wining strategy for Duplicator.
It is typically easier to demonstrate a winning strategy for Duplicator in \edashf games than in MS games, for several reasons.
First, it is easier to reason about only two structures at a time rather than about many structures at a time.
Second, in MS games, there is a tactic available to Spoiler (that is of no use in \edashf games) to make one move on top of an earlier move in one of the structures, and this can  greatly complicate the analysis.
On the other hand, in MS games,
Duplicator has the advantage of being able to make multiple copies of structures and  make different moves on the various  copies.
This feature can be very useful in proving a winning strategy for Duplicator.

A similar phenomenon of modifying the rules of the game to make it easier for Duplicator to win arose in defining and making use of Ajtai-Fagin games \cite{AjtaiFagin90} rather than making use of the originally defined Fagin games \cite{Fagin75} for proving inexpressibility results in monadic existential second-order logic  (called ``monadic NP" in \cite{FaginSV95}).
In Fagin games, there is a coloring round (a choice of the combinations of existentially-quantified monadic predicates) where Spoiler colors $A$ then Duplicator colors $B$, and then an ordinary \edashf game is played on the colored structures.
In Ajtai-Fagin games, Spoiler must commit to a coloring of $A$ without knowing what the other structure $B$ is.  
Fagin, Stockmeyer, and Vardi \cite{FaginSV95} use Ajtai-Fagin games to give a much simpler proof that connectivity is not in monadic NP than Fagin's original proof in \cite{Fagin75}.
In extending MS games to second-order logic, which we think is an interesting and important future step (and which is straightforward to define), these games can easily simulate Ajtai-Fagin games. 
This is because we can replace structure $A$ by the singleton set ${\mathcal A} = \{A\}$, and replace the structure $B$ by a set ${\mathcal B}$ that contains all possible choices for $B$ that Duplicator might choose in the Ajtai-Fagin game after Spoiler colors $A$.

\subsection{Related work} \label{sec:related_work}
Since our results can be viewed as giving information about the size of prefixes of sentences in prenex normal form, we begin by discussing some other papers that focus on such prefixes. 

Rosen \cite{Rosen2005} shows that there is a strict prefix hierarchy, based on the prefixes of sentences written in prenex normal form. The proof involves standard \edashf games.

Dawar and  Sankaran \cite{DaSa21} consider \edashf games, each of which focuses on a fixed prenex prefix.  For example, there is one game that deals with the prenex prefix $\exists \forall \exists$.
For each of these prefixes, they define an \edashf game on a pair $(A,B)$ of structures.  For example, in the $\exists \forall \exists$ game, Spoiler must move first in $A$, then in $B$, and then in $A$.  
Their Theorem 2.3 says that Spoiler has a winning strategy in a prefix game if and only if there is a sentence in prenex normal form  with exactly that prefix that is true about $A$ but not about $B$. Unfortunately, the ``only if” direction of their Theorem 2.3  is false \cite{Dawar21pc}.  This is because if $A$ is a linear order of size 5 and  $B$ is a linear order of size 4, and the prefix is $\exists \forall \exists$, then it turns out that Spoiler wins that 3-round prefix game, but it follows from  our Theorem~\ref{thm:g} that $A$ and $B$  agree on all sentences with at most three quantifiers, and in particular on all sentences $\exists x \forall y \exists z \phi(x,y,z)$, where $\phi$ is quantifier-free. 
Fortunately, in their paper, Dawar and  Sankaran just make use of the ``if” direction of Theorem 2.3, which is correct \cite{Dawar21pc}.

We now discuss some papers that, like ours, modify \edashf games by allowing a pair of sets of structures, rather than simply a pair of structures.
Adler and Immerman \cite{AdlImm03} use a type of \edashf game that involves a pair of sets of structures, where, as in our MS games, Duplicator can make multiple copies of structures and make different moves on them. 
Adler and Immerman’s concern is to obtain results about the size of sentences (rather than the number of quantifiers)  in transitive closure logic (first-order logic with the transitive closure operator). The rules of the game are rather complicated, since it must deal with transitive closure logic and capture the size of sentences. 

Hella and Vilander \cite{HelVil19} build on Adler and Immerman’s game, and their goal is also to determine sentence size (but in modal logic).  The rules of their game are also fairly complicated. 

 Grohe and Schweikardt \cite{GroheS05}
introduce a method (extended syntax trees) that corresponds to a game tree that is constructed by the two players in the Adler-Immerman game. They use these to study the size of sentences in the $2,3$ and $4$-variable fragments of first-order logic on linear orders.

Lotfallah \cite{Lotfallah04} introduces a class of E-F-like games played on a pair 
of sets of  structures rather than on a pair of single structures. 
In  Lotfallah’s games, Duplicator cannot make multiple copies of structures.
A follow-up paper by Lotfallah and Youssef \cite{LoYo05} characterizes certain first and second order 
prefix types but does not involve sets of structures. 

Hella and V{\"{a}}{\"{a}}n{\"{a}}nen \cite{HellaV15}, like Lotfallah, 
introduce a class of E-F-like games played on a pair
of sets  of structures rather than on a pair of single structures, where Duplicator cannot make multiple copies of structures. Hella and V{\"{a}}{\"{a}}n{\"{a}}nen use one variant of their game to
characterize 
the size of sentences needed for separating sets of structures in propositional logic and a second variant to characterize the size of sentences  needed for separating sets of structures in first-order logic. The first-order game is used for proving an exact bound on the size of existential sentences needed to define the length of linear orders.

\subsection{Overview of the Sections}
In Section~\ref{sec:fundamental} we prove the Equivalence Theorem~\ref{thm:main1}. 
In Section~\ref{sec:prelims} we establish certain preliminary terminology and notation and prove that the property of Duplicator having a winning strategy on two sets of structures gives rise to an equivalence relation over sets of structures with the same signature. The ensuing sections establish upper and lower bounds on the function $g(r)$ associated with Theorem~\ref{thm:g} until we are able to observe that we have tight bounds in all cases.
In Section~\ref{sec:upperbds} we establish upper bounds on $g(r)$, for $2$ and $3$. (The trivial tight bound $g(1) = 1$ is established in Section \ref{sec:prelims}.) The natural next step would be to proceed to higher values of $r$ using a type of recursive argument, but in Section~\ref{sec:interlude} we show why the natural recursive argument for MS games does not work. 
For pedagogical reasons, in Section \ref{sec:lower_bounds}, we jump to establishing lower bounds for $g(r)$. We then jump back to establishing upper bounds
in Section~\ref{sec:atoms}, where we introduce a new type of game, an \ms game with ``atoms",  which allows us to recurse and prove upper bounds for all $r$. 
The upper bounds are then seen to be tight with respect to our lower bounds and hence, in Section \ref{sec:final}, we are able to prove Theorem \ref{thm:g}. The machinery that we have built up then enables a quick proof of Theorem \ref{thm:fund_thm_for_los}, which is just the syntactic equivalent of Theorem \ref{thm:g_for_game_play} from the Introduction (there stated game theoretically). We note that while games with ``atoms" are an important part of our upper bound proofs, the final sentences guaranteed by 
Definition \ref{def:g} and Theorem \ref{thm:g}
do not contain atoms. 
We present our conclusions in Section \ref{sec:conc}.

\section{Proof of the Equivalence Theorem}\label{sec:fundamental}

To prove Theorem~\ref{thm:main1} we are going to add a special constant just prior to the play of each round that will help us maintain the induction. These constants, which we denote by $c_1,...,c_r$ are in addition to whatever constants may exist in the vocabulary $\tau$. We write $(A;\dots;c_1 \gets a_1)$ to mean the structure obtained after assigning $c_1$ to the element $a_1$ in $A$, and so on.

\begin{proof}[Proof of Theorem~\ref{thm:main1}]
Both directions are proved by induction on the number $r$ of rounds.  For $r = 0$, Spoiler winning in $0$ rounds means that for every $A \in \mathcal{A}$ and $B \in \mathcal{B}$, the restrictions $A',B'$ of $A$ and $B$ (respectively) to their constants must be non-isomorphic.  For every $A \in \mathcal{A}$, we can write a quantifier-free sentence $\phi_{A}$ that characterizes $A'$ up to isomorphism, using equality to identify any coinciding constants.  Then take $\phi$ to be the disjunction of all $\phi_{A}$ -- note that even if $\mathcal{A}$ is infinite, there are only finitely many distinct $\phi_{A}$.  Then $A \models \phi$ for all $A \in \mathcal{A}$.  Now consider any $B \in \mathcal{B}$.  We claim that $B \models \neg\phi$.  As $\neg\phi$ is a conjunction, this is equivalent to $B \models \neg\phi_{A}$ for every $A$.  Suppose not, then we would have $B \nvDash \neg\phi_{A}$, i.e., $B \models \phi_{A}$.  But because $\phi_{A}$ characterizes $A'$ up to isomorphism, this would make $B'$ isomorphic to $A'$, contradicting that Spoiler wins.  Hence $\phi$ is a quantifier-free sentence that distinguishes $\mathcal{A}$ and $\mathcal{B}$. 

Conversely, if there is a quantifier-free sentence $\phi$ that distinguishes $\mathcal{A}$ and $\mathcal{B}$, then there cannot exist $A \in \mathcal{A}$ and $B \in \mathcal{B}$ such that the restrictions $A'$ and $B'$ to the constants appearing in $\phi$ are isomorphic.  Thus, Duplicator loses without further play. 

Now suppose $r \geq 1$ and the equivalence is true for $r-1$.  We induct on the first move rather than the last move of the games.  For the forward direction, suppose Spoiler can win---say by playing in $\mathcal{B}$.  For each $B \in \mathcal{B}$, Spoiler selects an element $b \in B$.  Duplicator replies by replicating every $A \in \mathcal{A}$ and playing every possible $a \in A$. Now let us assign the constant $c_1$ to the element played on each of the structures in $\mathcal{A}$ and $\mathcal{B}$.  The resulting game position $(\mathcal{A}^1,\mathcal{B}^1)$, which includes the new constant $c_1$ in each structure, is winnable in $r-1$ rounds by Spoiler.  By the induction hypothesis, there is a sentence $\psi$ with $r-1$ quantifiers that distinguishes $\mathcal{B}^1$ from $\mathcal{A}^1$.  Now define $\phi = (\exists x_1)\psi'$ where $\psi'$ replaces all occurrences of $c_1$ in $\psi$ by $x_1$.
(This is alright even in degenerate cases where $c_1$ does not occur in $\psi$.)
For every $B \in \mathcal{B}$, we have $B \models \phi$ because there is a $b \in B$ such that $(B;\dots;c_1 \gets b) \models \psi$, namely the $b$ that Spoiler played in $B$.  Hence it suffices to show that $A \models \neg\phi = (\forall x_1)\neg\psi'$ for every $A \in \mathcal{A}$.  After Duplicator's play, $A$ was replaced by
\[
\{(A;\dots;c_1 \gets a_1),(A;\dots;c_1 \gets a_2),\dots,(A;\dots;c_1 \gets a_m)\},
\]
where $A = \{a_1,\dots,a_m\}$.  Since $\psi$ distinguishes $\mathcal{B}^1$ from $\mathcal{A}^1$, we have $(A;\dots,c_1 \gets a_j) \models \neg\psi$ for each $j = 1,\dots,m$.  It follows that $A \models (\forall x_1)\neg\psi'$.  The case where Spoiler wins by playing in $\mathcal{A}$ is handled symmetrically.  

Going the other way, suppose $\phi$ is a prenex sentence with $r$ quantifiers that distinguishes $\mathcal{A}$ from $\mathcal{B}$.  If the leading quantifier is $\forall$ then $\neg\phi$ has leading quantifier $\exists$ and distinguishes $\mathcal{B}$ from $\mathcal{A}$, so we can reason by symmetry.  So let $\phi = (\exists x_1)\psi'$ for some $\psi'$ and take $\psi$ to be the sentence $\psi'$ with $x_1$ replaced everywhere by the special constant symbol $c_1$.  For every $A \in \mathcal{A}$, $A \models \phi$, so there exists $a_1 \in A$ such that $(A;\dots;c_1 \gets a_1) \models \psi$.  Spoiler can play such an element $a_1$ in every $A$.  Now every $B \in \mathcal{B}$ models $\neg\phi = (\forall x_1)\neg\psi'$.  For every $b \in B$, Duplicator creates the structure $(B;\dots;c_1 \gets b)$, but regardless of $b$, it models $\neg\psi$.  Thus, $\psi$ distinguishes the resulting set $\mathcal{A}^1$ from Duplicator's $\mathcal{B}^1$, has $r-1$ quantifiers, and includes $c_1$ along with any previous constants.  By the induction hypothesis, Spoiler wins from $(\mathcal{A}^1,\mathcal{B}^1)$ in $r-1$ rounds, so Spoiler wins from $(\mathcal{A},\mathcal{B})$ in $r$ rounds.
\end{proof}

\begin{proof}[Proof of Theorem~\ref{thm:main1a}]
In both directions of the induction, if the first quantifier $Q_1$ is fixed to be $\exists$ or fixed to be $\forall$, then this limits which case can arise within the induction step, but does not affect its validity. 
\end{proof}

\section{Preliminaries} \label{sec:prelims}

\begin{defi}
Let $\mathcal{A} = \{A\}$ and $\mathcal{B} = \{B\}$ be two \textit{singleton sets of structures}. Write $\mathcal{A} \equiv_r \mathcal{B}$ iff Duplicator has a winning strategy for \ms games of $r$ rounds on $\mathcal{A}$ and $\mathcal{B}$. 
\end{defi}

An important consequence of Theorem \ref{thm:main1} is the following. 

\begin{lem} \label{lemma:prenex_equiv} The relation $\equiv_r$ is an equivalence relation between singleton sets of structures.
\end{lem}

\begin{proof} That the relation $\equiv_r$ is reflexive and symmetric follows immediately from the definition. For transitivity, suppose there are three singleton sets of structures, $\mathcal{A} = \{A\}, \mathcal{B} = \{B\}$ and $\mathcal{C} = \{C\}$ such that $\mathcal{A} \equiv_r \mathcal{B}$ and $\mathcal{B} \equiv_r \mathcal{C}$.  By the Equivalence Theorem \ref{thm:main1}, $\mathcal{A} \equiv_r \mathcal{B}$ implies that $\mathcal{A}$ and $\mathcal{B}$ agree on the same sentences with at most $r$ quantifiers. Similarly, $\mathcal{B} \equiv_r \mathcal{C}$ implies that $\mathcal{B}$ and $\mathcal{C}$ agree on the same set of sentences with at most $r$ quantifiers. Hence, $\mathcal{A}$ and $\mathcal{C}$ agree on these same set of sentences. By the Equivalence Theorem again, it follows that $\mathcal{A} \equiv_r \mathcal{C}$.
\end{proof}

\begin{observation} \label{obs:phokion}
Among non-singleton sets of structures, the property of there being a Duplicator-winning strategy between two sets does \textit{not} induce an equivalence relation. 
\end{observation}
The induced relation between pairs of sets of structures can fail to be transitive. To see this, let $A$ be a graph with no self-loops among the nodes, and let $C$ be a graph with at least one node with a self-loop. Now let $\mathcal{A} = \{A\}, \mathcal{B} = \{A, C\}$, and $\mathcal{C} = \{C\}$. Duplicator wins a 1-round M-S game on $\mathcal{A}$ and $\mathcal{B}$ simply by focusing on $A$, and, analogously, Duplicator wins a 1-round M-S game on $\mathcal{B}$ and $\mathcal{C}$ simply by focusing on $C$.  However, Spoiler wins the 1-round M-S game on $\mathcal{A}$ and $\mathcal{C}$ by playing on a node in $C$ that has a self-loop.

We note that Definition 3.1 and Lemma 3.2  are incorrectly stated for sets of structures in both the LICS conference \cite{Fagin21} and extended ArXiV \cite{Fagin21a} versions of this manuscript.  The fallacious bit of reasoning is that in the case of non-singleton sets $\mathcal{A}, \mathcal{B}$ of structures, the fact that Duplicator has a winning strategy in an $r$-round M-S game does \textit{not} imply that all elements of $\mathcal{A}$ and $\mathcal{B}$ agree on every $r$-quantifier sentence, as was claimed -- consider, for example, the sets $\mathcal{A}$ and $\mathcal{B}$ in the preceding paragraph. We thank Phokion Kolaitis for pointing out this error to us.

\medskip

Before proceeding further let us establish some terminology that is intended to make the reading smoother.  In cases where there are multiple linear orders on one side or another we often refer to the different linear orders as different ``boards'' that Spoiler and Duplicator play on.
 
When we say a ``$K$ versus $K'$ game'' or a ``$K$ vs.\ $K'$ game'', we mean an \ms game played on $(\mathcal{A}, \mathcal{B})$ where $\mathcal{A}$ consists of a single linear order of size $K$, and $\mathcal{B}$ consists of a single linear order of size $K'$.

We will typically play games where $\mathcal{A}$ and $\mathcal{B}$ each consist of a single linear order as above. In this context, as we did in the Introduction, we will use $B$ to denote the \textit{big} linear order and $L$ to denote the \textit{little} linear order.

As is standard in model theory, we assume a non-empty universe so all linear orders are of size at least $1$. With reference to the function $g$ in Definition \ref{def:g}, we begin by establishing:

\begin{lem} \label{lemma:g_of_1}
$g(1) = 1$.
\end{lem}

\begin{proof}
The sentence $\exists x(x=x)$ is true for all linear orders so that $g(1) \geq 1$.  Moreover, Duplicator can win $1$-round \ms games whenever both linear orders are of size $1$ or greater, which implies that $g(1) \leq 1$.
\end{proof}

\section{Towards Establishing Upper Bounds on $g(r)$}\label{sec:upperbds}

A potent tool for finding an upper bound $k$ on the value of $g(r)$ will be to find strategies such that Duplicator can win $r$-round games on a \textit{pair of singleton linear orders} whenever the sizes of the linear orders are at least $k$. All of our upper bounds are established in this manner. 

Since we have established that $g(1) = 1$, 
we start establishing upper bounds at $g(2)$.

\begin{lem} \label{lemma:g_upper_of_2}
Duplicator can win $2$-round \ms games whenever both linear orders are of size $2$ or greater, and hence $g(2) \leq 2$.
\end{lem}

\begin{proof}
In the Introduction we considered a $2$-round \ms game on two linear orders of sizes $|B| = 3, |L| = 2$. 
Figure \ref{fig:3_vs_2_singleton-intro} is given again here for ease of reference.
\begin{figure} [ht]
\centerline{\scalebox{0.40}{\includegraphics{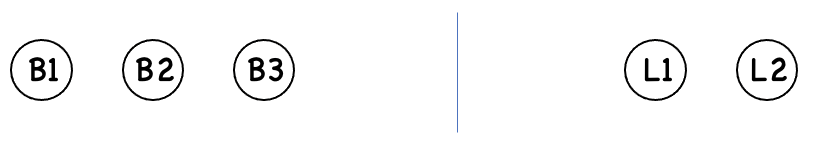}}}
\caption{The case $|B| = 3, |L| = 2$.}
\label{fig:3_vs_2_singleton}
\end{figure}
The Introduction covered the case where Spoiler selects the middle $B$ element, $B(2)$, in the first round. The case where Spoiler picks an end element from either the $L$ or $B$ boards in the first round is easier -- Duplicator just picks the corresponding end element from the other linear order and she does not even need to make a second copy of the board to win. For example, in response to $B(1)$, Duplicator will play $L(1)$, or in response to $L(2)$, Duplicator will play $B(3)$, in either case leading to simple wins.

In the example given in the Introduction, where Spoiler played $B(2)$, once Duplicator makes copies and plays different moves in each copy, we render the game after round 1, as in Figure \ref{fig:3_vs_2_rd1}.
\begin{figure} [ht]
\centerline{\scalebox{0.40}{\includegraphics{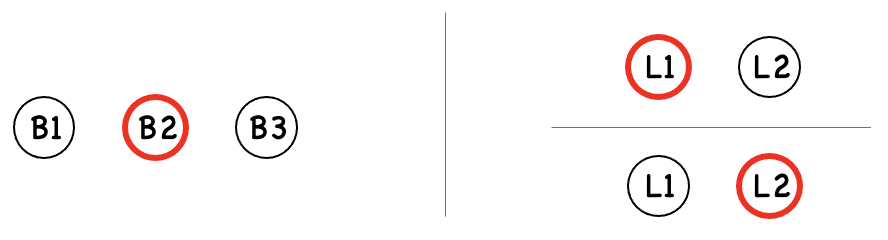}}}
\caption{In response to Spoiler playing $B(2)$, Duplicator makes a second copy of the $L$ board and place $L(1)$ on one board and $L(2)$ on the other board.  (The diagrams omit parentheses.)}
\label{fig:3_vs_2_rd1}
\end{figure}

\medskip

To complete the analysis, we must show that Duplicator wins 
whenever $2 \leq |L| < |B|$. Let us begin with the case $|L| = 2$ and $|B| > 3.$
 If Spoiler picks an end element from $B$ or from $L$, then Duplicator picks the corresponding end element from the opposite side and wins. On the other hand, if Spoiler picks a non-end element on $B$, Duplicator picks $L(1)$ on top and $L(2)$ on bottom, winning just like in the introduction. We are left to consider the case when $3 \leq |L| < |B|$. Playing end moves from either $B$ or $L$ have the same effect as before. While if Spoiler picks a non-end element from either $B$ (or $L$), Duplicator wins by playing \textit{any} non-end element from, respectively, $L$ (or $B$), guaranteeing a $2$-round win. 
\end{proof}

\begin{lem} \label{lemma:g_upper_of_3}
Duplicator can win $3$-round \ms games on linear orders whenever both linear orders are of size $4$ or greater, and hence $g(3) \leq 4$.
\end{lem}

\begin{proof}
Let us start with the base case $|B|=5, |L|=4$. See Figure \ref{fig:5_vs_4_base}. 
\begin{figure} [ht]
\centerline{\scalebox{0.30}{\includegraphics{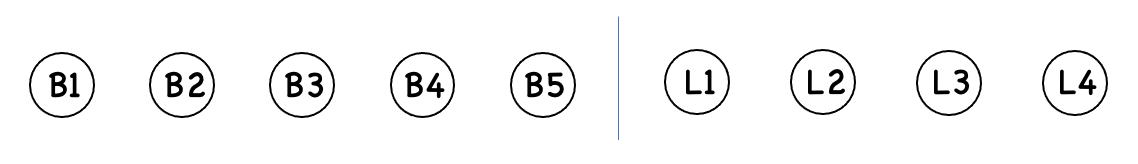}}}
\caption{The case $|B| = 5, |L| = 4$.}
\label{fig:5_vs_4_base}
\end{figure}
Duplicator-winning outcomes associated with all Spoiler 1st round plays, \textit{other than} $B(3)$, are easy to analyze and described in the table below. Further explanation of the notation used in the table is given in the text that follows it.

\medskip

\begin{small}
\begin{center}
\begin{tabular}
{|cc|cc|cc|}
\hline
\multicolumn{2}{|c|}{\textbf{Round 1}} & \multicolumn{2}{|c|}{\textbf{Round 2}} & \multicolumn{2}{|c|}{\textbf{Round 3}} \\
\hline
\textbf{S} & \textbf{D} & \textbf{S} & \textbf{D} & \multicolumn{2}{|l|}{}\\
\hline
$B(1)$ & $L(1)$ & \multicolumn{2}{|l|}{D wins by reduction to $f(2)$} & \multicolumn{2}{|l|}{} \\
$B(2)$ & $L(2)$ & $B(1)$ & $L(1)$ & \multicolumn{2}{|l|}{D wins} \\
 &  & $B(2)$ & $L(2)$ & \multicolumn{2}{|l|}{D wins} \\
 &  & $B(3)$ & $L(3)$ & \multicolumn{2}{|l|}{D wins} \\
 &  & $B(4)$ & $L(3)$ & \multicolumn{2}{|l|}{D wins on this board or board below} \\
 &  & & $L(4)$ & \multicolumn{2}{|l|}{} \\
 &  & $B(5)$ & $L(4)$ & \multicolumn{2}{|l|}{D wins} \\ 
 &  & $L(1)$ & $B(1)$ & \multicolumn{2}{|l|}{D wins by transposition to Rd2. $B(1)$,$L(1)$} \\ 
 &  & $L(2)$ & $B(2)$ & \multicolumn{2}{|l|}{D wins by transposition to Rd2. $B(2)$,$L(2)$} \\ 
 &  & $L(3)$ & $B(3)$ & \multicolumn{2}{|l|}{D wins by transposition to Rd2. $B(3)$,$L(3)$}  \\ 
 &  & $L(4)$ & $B(5)$ & \multicolumn{2}{|l|}{D wins by transposition to Rd2. $B(5)$,$L(4)$}  \\
 $B(3)$ & $L(1)$ & \multicolumn{2}{|l|}{\textbf{Described in the text}}  & \multicolumn{2}{|l|}{} \\
  & $L(2)$ & \texttt{"} & \texttt{"} & \multicolumn{2}{|l|}{} \\
  & $L(3)$ & \texttt{"} & \texttt{"} & \multicolumn{2}{|l|}{} \\
  & $L(4)$ & \texttt{"} & \texttt{"} & \multicolumn{2}{|l|}{} \\
 $B(4)$ & $L(3)$ & \multicolumn{2}{|l|}{Symmetrical to Rd1. $B(2)$,$L(2)$}  & \multicolumn{2}{|l|}{} \\
 $B(5)$ & $L(4)$ & \multicolumn{2}{|l|}{Symmetrical to Rd1. $B(1)$,$L(1)$}  & \multicolumn{2}{|l|}{} \\
 $L(1)$ & $B(1)$ & \multicolumn{2}{|l|}{Transposition to Rd1. $B(1)$,$L(1)$}  & \multicolumn{2}{|l|}{} \\
 $L(2)$ & $B(2)$ & \multicolumn{2}{|l|}{Transposition to Rd1. $B(2)$,$L(2)$}  & \multicolumn{2}{|l|}{} \\
 $L(3)$ & $B(4)$ & \multicolumn{2}{|l|}{Transposition to Rd1. $B(4)$,$L(3)$}  & \multicolumn{2}{|l|}{} \\
 $L(4)$ & $B(5)$ & \multicolumn{2}{|l|}{Transposition to Rd1. $B(5)$,$L(4)$}  & \multicolumn{2}{|l|}{} \\
\hline
\end{tabular}
\end{center}
\end{small}

A few notes on the table:
Moves given in the respective \textbf{S} columns are Spoiler plays while those given in the \textbf{D} columns are Duplicator plays.  When we say that a game is winnable for Duplicator ``by reduction to $f(k)$,'' for some $k$, we mean that the game from this point on can be won by Duplicator simply as a $1$-board standard \ef game. Recall the definition of $f$ from Theorem \ref{thm:f}. In these cases, typically an end element has been played, and from an \ef point of view, it is of no benefit to Spoiler to subsequently play on top of the 1st element, so we may remove these elements from both $L$ and $B$ and consider the game, in subsequent rounds, to be played solely on the remaining elements. The game reduces to $f(k)$ if there are $k$ rounds yet to be played and both sets of remaining elements are at least of size $f(k)$. 

When we say that a given sequence of moves $X,Y$ is ``symmetrical'' to another sequence of moves $X', Y'$, we mean that you can arrive at one sequence by taking the mirror image of the other sequence with respect to the center of the boards, where the other sequence has already been analyzed.

We say that a sequence of moves is a ``transposition'' of another sequence of moves if the end board positions are the same but the sequence of moves leading to that board position is different.

Note that if Duplicator is able to win in a single additional move by focusing just on a single board on the $B$ and $L$ sides without any special strategy, we do not describe every possible move and response sequence (there are just too many), but rather just indicate that ``D wins.''

In response to the one tricky Spoiler 1st round move of $B(3)$, Duplicator creates additional copies of the $L$ board and makes every possible move, as depicted in Figure \ref{fig:5_vs_4_rd1}.
\begin{figure} [ht]
\centerline{\scalebox{0.35}{\includegraphics{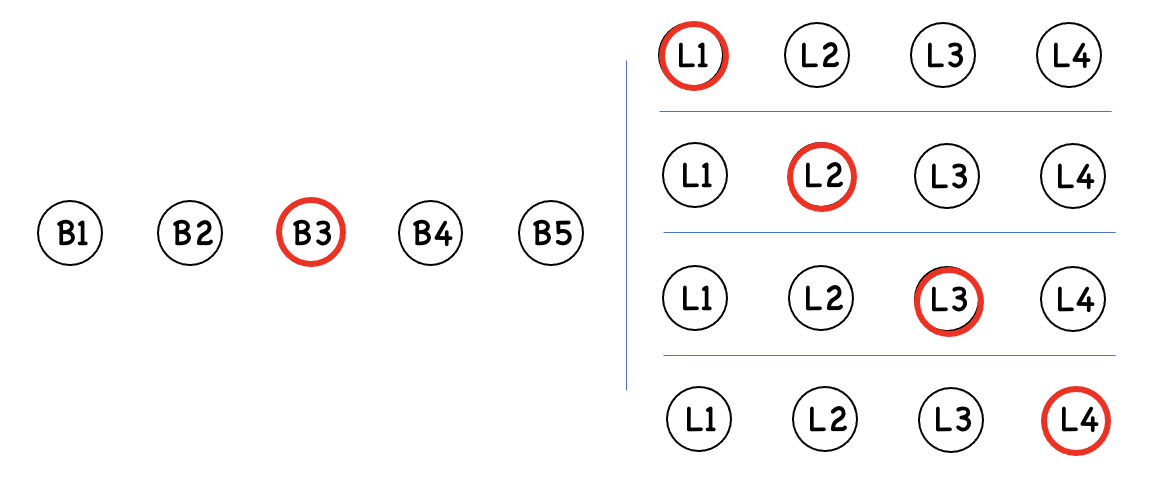}}}
\caption{After Spoiler plays $B(3)$ from the $B$ side, Duplicator plays $L(1)$, $L(2)$, $L(3)$ and $L(4)$ on different boards on the $L$ side.}
\label{fig:5_vs_4_rd1}
\end{figure}
Let's now analyze the possible Spoiler 2nd round responses. 

\smallskip

\noindent \underline{Case where Spoiler makes 2nd round move on $B$}: 
\begin{small}
\begin{center}
\begin{tabular}
{|cc|}
\hline
\multicolumn{2}{|c|}{\textbf{Round 2}} \\
\hline 
\textbf{Spoiler} & \textbf{Duplicator} \\
\hline
$B(1)$ & $L(1)$ on 3rd $L$ board insuring a win on that board \\
$B(2)$ & $L(2)$ on 3rd $L$ board insuring a win on that board \\
$B(3)$ & $L(3)$ on 3rd $L$ board insuring a win on that board \\
$B(4)$ & $L(3)$ on 2nd $L$ board insuring a win on that board \\
$B(5)$ & $L(4)$ on 2nd $L$ board insuring a win on that board \\
\hline
\end{tabular}
\end{center}
\end{small}

\bigskip

\noindent  \underline{Case where Spoiler makes 2nd round moves on $L$}:

\smallskip

First consider the possibility of Spoiler playing atop an existing move. If he plays atop $L(2)$ on board 2, or $L(3)$ on board 3, then Duplicator will play atop $B(3)$, ensuring victory on the associated boards in another move. 
Analogously, if Spoiler plays on top of \textit{both} $L(1)$ on the top board and $L(4)$ on the bottom board, then Duplicator can play atop $B(3)$ guaranteeing a win on one of these $L$ boards or the other. Thus, we may assume Spoiler plays atop \textit{at most} one of the existing moves, with that one move being atop either $L(1)$ or $L(4)$. 
Of the at least three boards in which Spoiler does \textit{not} play on top of existing moves, the moves are either to the right of the existing moves, or to the left of the existing moves, and hence, there must be at least two played to one side or the other of the existing moves.
Without loss of generality, assume that at least two of these moves are to the right of the existing moves. 

Suppose that one of the moves to the right is on the 2nd board. An $L(3)$ move would be met by a $B(4)$ response by Duplicator, while an $L(4)$ move would be met by a $B(5)$ response, in either case leading to a single board win for Duplicator. Thus, we can assume that the two 2nd round moves to the right of the existing moves are on the 1st and 3rd boards -- the only two boards we need to consider to conclude this bit of the analysis.  The element $L(4)$ must be selected from board 3. If $L(2)$ is selected from board 1 then Duplicator wins by selecting $B(4)$: any 3rd round Spoiler move on $B$ is parried on one of the two $L$ boards, while if Spoiler plays his 3rd round from $L$, Duplicator can maintain an isomorphism with any move played on either the 1st or 3rd board.  We are left to consider just the possibilities that Spoiler plays $L(3)$ or $L(4)$ for his 2nd round move on board 1. Suppose he plays $L(3)$. Duplicator can then win by playing $B(5)$: if Spoiler plays his 3rd move from $B$ then $B(4)$ is met with $L(2)$ on the 1st board and any other $B$ move is easily parried on the 3rd $L$ board. On the other hand, if Spoiler plays his 3rd move from $L$ then whatever he does on the 3rd $L$ board can be matched with an isomorphism-preserving move on $B$. In the final case where Spoiler plays $L(4)$ for his 2nd round move on the 1st $L$ board, Duplicator responds with $B(5)$, guaranteeing an isomorphism. It follows that Duplicator can always win the $|B| = 5, |L| = 4$, $3$-round game.

To complete the argument we must show that Duplicator can also win when one or both of $|B| > 5$ and $|L| > 4$. Let us start by considering the case of $|B| > 5, |L| = 4$. The analysis for any initial Spoiler move where he plays either an end move or a next-to-end move is precisely as earlier. Consider all other moves on $B$ to be ``middle'' moves. Any such middle move is parried just like in Figure \ref{fig:5_vs_4_rd1}, by playing $L(1)$, $L(2)$, $L(3)$, and $L(4)$ on the different boards on the $L$ side. The only new wrinkle in this analysis occurs in the case where Spoiler first plays a middle move on $B$. We illustrate for $|B| = 6$ and where Spoiler's first play is $B(3)$. He can now play $B(5)$ and Duplicator must take new precautions because she cannot win by playing on a single board. She can, however, play $L(3)$ on the first $L$ board, covering further play on the right of $B(3)$ by Spoiler, while also playing, say, $L(3)$ on board 2, to cover potential play by Spoiler on the left of $B(3)$. The rest of the analysis is precisely as in the smaller $|B|$ case. 

Finally, consider the case $|B| > |L| > 4$. If Spoiler plays an end element or next-to-end element, Duplicator follows suit, playing, respectively, an end element or next-to-end element from the same side on the opposite board, leading to an easy victory: if Spoiler plays on both sides of the 1st round move in subsequent rounds he will clearly lose, while playing just on one side reduces to a losing $2$ round game for him. Similarly, if Spoiler plays a ``middle'' element, Duplicator can respond playing any middle element from the opposite side. Again, if Spoiler plays on both sides of the 1st round move in subsequent rounds he will clearly lose, but playing just on one side reduces to a losing $2$ round game for him. 

That completes the argument that Duplicator can win an \ms game of $3$ rounds whenever the size of all boards is at least $4$ and hence establishes the lemma. 
\end{proof}

\medskip

At this point it is natural to suspect that one can build up upper bounds recursively in a relatively simple manner. However, such an approach runs into unexpected difficulties, as we describe in the next section.

\section{Interlude: Why Naive Recursion cannot be used\texorpdfstring{\\}{} to Build up Duplicator Winning Strategies in \MS Games} \label{sec:interlude}

It is worth pausing to understand why a simple idea to use recursion to build up Duplicator-winning strategies, and hence upper bounds on $g(r)$, fails. Via Lemma \ref{lemma:g_upper_of_3}, we have established the fact that $g(3) \leq 4$ by showing that Duplicator wins $3$-round \ms games if the sizes of both linear orders are $4$ or larger. It is tempting to try to use this fact to produce a Duplicator strategy for winning $4$-round games using recursion. To understand the problematic logic, it suffices to consider a $4$-round game on boards of sizes $9$ and $10$.  The erroneous argument runs as follows. Suppose Spoiler plays $L(5)$ on his 1st move. Duplicator can then simply reply with the single move $B(5)$ (so the erroneous reasoning goes), as in Figure \ref{fig:10_vs_9_recursion_attempt}.
\begin{figure} [ht]
\centerline{\scalebox{0.33}{\includegraphics{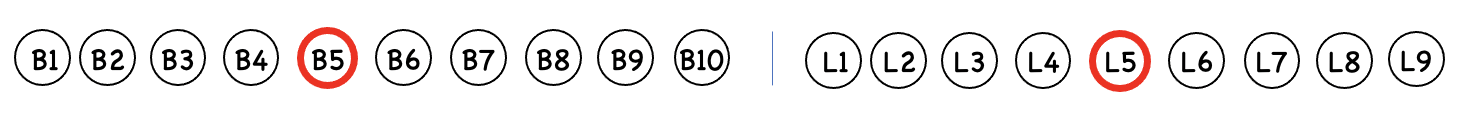}}}
\caption{A simple attempt to arrive at a Duplicator-winning strategy for the $10$ versus $9$ game.  The first round moves are given in red.}
\label{fig:10_vs_9_recursion_attempt}
\end{figure}
Since there are five unplayed elements to the right of $B(5)$ and four unplayed elements to the right of $L(5)$, Duplicator should now just be able to mimic Spoiler's moves at, or to the left of, $B(5)$/$L(5)$ and otherwise play moves to the right of $B(5)$/$L(5)$ as if it were a $3$-round, $5$ versus $4$ game, which we know is winnable by Duplicator. In fact, as we will learn in Lemma~\ref{lemma:g_underline_of_4}, the $10$ \vs $9$ game is winnable by \textit{Spoiler}. Hence this strategy does not work. 

The reason the strategy doesn't work is that there is interaction between play on the two sides and there are moves to the left of the $5$ \vs $4$ sub-game that are more powerful for Spoiler (in the sense of breaking more to-that-point maintained partial isomorphisms) than any moves available in the $5$ \vs $4$ game. This additional power is achieved by Spoiler playing atop an already played move that is not part of the $5$ \vs $4$ sub-game at a critical juncture. 

Indeed, for his 2nd round move, Spoiler will select $B(8)$, and as we know from the analysis of the $5$ \vs $4$ games, this will require Duplicator to make copies of the $L$ board and play each of the possible moves to the right of $L(5)$, as depicted in Figure \ref{fig:10_vs_9_recursion_attempt2}.
\begin{figure} [ht]
\centerline{\scalebox{0.40}{\includegraphics{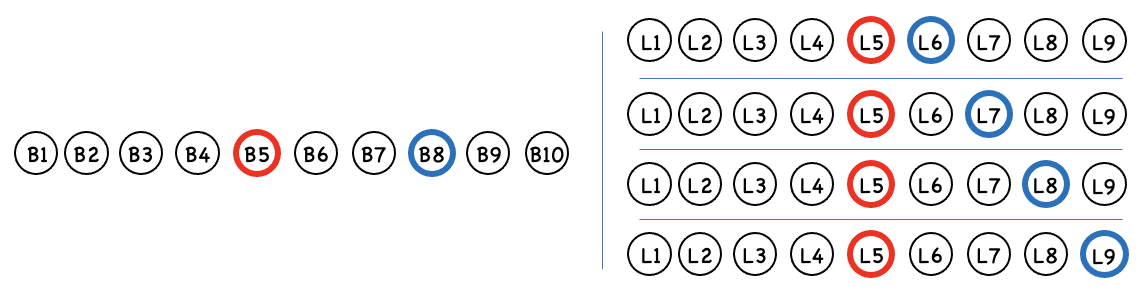}}}
\caption{The 2nd round plays of the simple Duplicator-winning strategy for the $10$ versus $9$ game.  The first round moves are given in red, 2nd round moves in blue.}
\label{fig:10_vs_9_recursion_attempt2}
\end{figure}
At this point Spoiler will play on top of $L(5)$ on the 1st $L$ board (a move that \textit{wasn't available} in the $5$ \vs $4$ game), on $L(6)$ on the 2nd $L$ board, on $L(9)$ on the 3rd $L$ board, and on top of $L(9)$ on the 4th $L$ board. See Figure \ref{fig:10_vs_9_recursion_attempt3a}.
\begin{figure} [ht]
\centerline{\scalebox{0.40}{\includegraphics{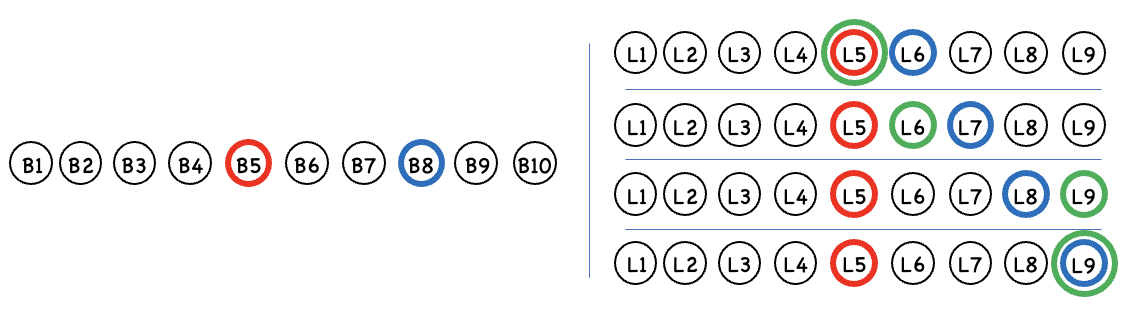}}}
\caption{3rd round plays (in green) that foils the simple Duplicator-winning strategy for the $10$ versus $9$ game.  The first round moves are given in red, 2nd round moves in blue.}
\label{fig:10_vs_9_recursion_attempt3a}
\end{figure}
Spoiler is going to play his 4th and final round moves from $B$, but first let's see what happens in response to the various Duplicator 3rd round moves (we assume the oblivious strategy). The only move that would keep an isomorphism with the 1st $L$ board is a play atop $B(5)$ -- breaking isomorphisms with any $L$ board but the top one. But then Spoiler will play $B(6)$ on his 4th move on the same board, and Duplicator will then not be able to keep the isomorphism going with the top $L$ board. On the other hand, to maintain an isomorphism with the 2nd $L$ board, Duplicator must play either $B(6)$ or $B(7)$, and in so doing, break an isomorphism with any other $L$ board. However, Spoiler will respond with $B(7)$ if $B(6)$ was played, or $B(6)$ if $B(7)$ was played, and Duplicator will have no retort. To maintain an isomorphism with the 3rd $L$ board, Duplicator must play $B(9)$ or $B(10)$, again breaking the isomorphism with all other $L$ boards, and Spoiler will respond by playing $B(10)$ if $B(9)$ was played and vice versa. Finally, to keep an isomorphism going with the bottom $L$ board, Duplicator must play on top of $B(8)$, again breaking the isomorphisms with other $L$ boards. But then Spoiler plays $B(9)$ and Duplicator cannot respond.

Thus, trying to replicate the $5$ \vs $4$ strategy to the right of $B(5)$/$L(5)$ and mimicking play on top of or to the left of $B(5)$/$L(5)$ does not work for Duplicator. 
When Spoiler played on top of $L(5)$ for one of his 3rd round moves, he was utilizing a move that was \textit{not} available to him in the $5$ \vs $4$ game. 
He \textit{did} have the option to play on top of $L(6)$ (which was labeled $L(1)$ in the $5$ \vs $4$ game) -- but doing so would have duplicated the isomorphism type of the bottom board. Playing on top of $L(5)$ has stronger effect since if Duplicator is to maintain isomorphisms with both the top and bottom boards, she cannot do it on a single board on the left, she must make a copy and maintain an isomorphism with the top-right board in one copy, and maintain an isomorphism with the bottom-right board in the other copy. 

Before moving on, it is worth noting that it was not strictly necessary for Spoiler to play on top of $L(5)$ on the 1st $L$ board on his 3rd move -- any move on $L(1)$--$L(4)$ would have worked just as well. See Figure \ref{fig:10_vs_9_recursion_attempt3b}.
\begin{figure} [ht]
\centerline{\scalebox{0.40}{\includegraphics{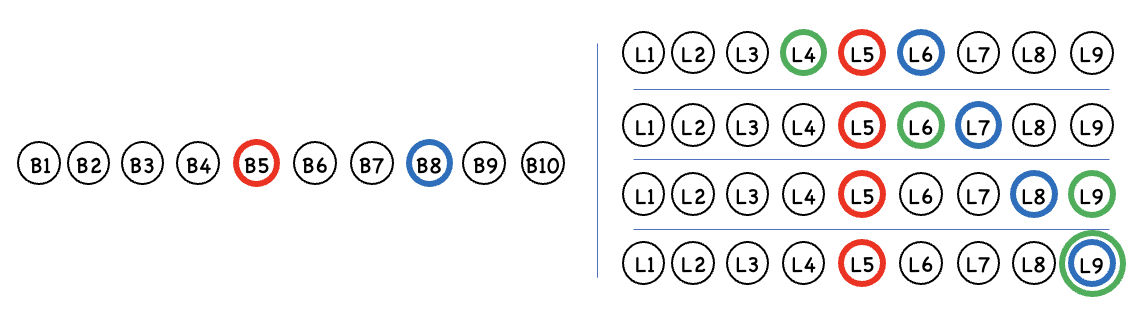}}}
\caption{A second example of 3rd round plays (in green) that foils the simple Duplicator-winning strategy for the $10$ versus $9$ game.  The first round moves are given in red, 2nd round moves in blue.}
\label{fig:10_vs_9_recursion_attempt3b}
\end{figure}
For Duplicator to maintain an isomorphism with the top $L$ board she will have to move on $B(1)$--$B(4)$, again breaking the isomorphisms with any other $L$ boards, and Spoiler can then pick $B(6)$ on his 4th move, 
breaking any hope for an isomorphism with the 1st board. 

We remark that every winning strategy for Spoiler in the 10 vs.\ 9 game requires at least one play-on-top move. We show this at the end of Appendix \ref{sec:app-g_underline_of_4}, as summarized in Observation \ref{obs:play-on-top}.

\section{Lower Bounds on $g(r)$} \label{sec:lower_bounds}

The upshot of the prior section is that establishing tight bounds for $g(r)$ requires careful attention.  Having established $g(1) = 1$ (Lemma \ref{lemma:g_of_1}), we begin our consideration of lower bounds with $g(2)$.

\begin{lem} \label{lemma:g_lower_of_2}
$g(2) \geq 2$
\end{lem}

\begin{proof}
The sentence $\Phi_2 = \exists x\exists y(x < y)$ distinguishes linear orders of size $2$ and above, from the linear order of size $1$.
\end{proof}

\begin{lem} \label{lemma:g_lower_of_3}
$g(3) \geq 4$
\end{lem}

\begin{proof}
The following sentence, with $3$ quantifiers, 
distinguishes linear orders of size at least $4$ from those of size 
at most $3$:
\begin{equation}
    \Phi_3 = \forall x\exists y \exists z(x < y < z \vee y < z < x)\qedhere  \label{g3_forall} 
\end{equation} 
\end{proof}

\begin{lem} \label{lemma:g_underline_of_4}
${g}(4) \geq 10$.
\end{lem}

\begin{proof}
 The following sentence with $4$ quantifiers distinguishes linear orders of size at least $10$ from those of size at most $9$.

\begin{small}
\begin{align}
	\Phi_4 = \forall x \exists y \forall z \exists w(& \notag \\
		&x < z < y \rightarrow (w \neq z \land x < w < y) ~~~\land \label{cond1} \\
		&x < y < z \rightarrow (w \neq z \land x < y < w) ~~~\land \label{cond2}\\
		&y < z < x \rightarrow (w \neq z \land y < w < x)~~~ \land\label{cond3} \\
		&z < y < x \rightarrow (w \neq z \land w < y < x)~~~\land \label{cond4}\\
		&z = x \rightarrow (x < w < y \vee y < w < x)~~~\land \label{cond5}\\
		&z = y \rightarrow (x < y < w \vee w < y < x)). \label{cond6}
\end{align}
\end{small}

This sentence captures the fact that ``for every $x$ there is a $y$ with two or more elements on each side of $y$, both of which are on the same side of $x$ as $y$'', a fact that is true for linear orders of size $10$ or greater, but not for linear orders of size less than $10$. For example, in a linear order of size $9$, the middle element will not have an element to either side of it having these properties. More specifically, the first four implications (\ref{cond1})--(\ref{cond4}) say that for every $x$ there is a $y$ such that for any $z \neq x,y$, with $z$ on the same side of $x$ as $y$, there is an additional element besides $z$ on the same side of $y$ and also on the same side of $x$ as $y$. The last two implications (\ref{cond5})--(\ref{cond6}) are critical and insure that there are elements (i) between $x$ and $y$ and (ii) less than $y$ if $y < x$, and greater than $y$ if $y > x$.
\end{proof}

For an alternate, game-based proof of Lemma \ref{lemma:g_underline_of_4} see Appendix \ref{sec:app-g_underline_of_4}.

\medskip

With initial values $g(1) = 1, g(2) \geq 2$ (Lemma \ref{lemma:g_lower_of_2}), $g(3) \geq 4$ (Lemma \ref{lemma:g_lower_of_3}) and $g(4) \geq 10$ (Lemma \ref{lemma:g_underline_of_4}), we establish all remaining lower bounds via a game argument:

\begin{thm} \label{thm:main_lower}
For $r > 4$, 
\begin{equation} \label{eqn:main_recursion_lower}
	g(r) \geq \begin{cases} 2g(r-1)~~~~~~\textrm{if $r$ is even,}\\ 2g(r-1) + 1\textrm{ if $r$ is odd.} \end{cases}
\end{equation}
\end{thm}

\begin{proof}
Let us establish (\ref{eqn:main_recursion_lower}) first for odd $r$.  Here we need to provide a Spoiler-winning strategy for a game with arbitrary linear orders of sizes at least $2g(r-1)+1$ on one side and  arbitrary linear orders of sizes at most $2g(r-1)$ on the other side. For simplicity, we will show the strategy for the game in which we have a single linear order of size at least $2g(r-1)$ + 1 on one side and a single linear order of size at most $2g(r-1)$ on the other side. The strategy for the case of multiple linear orders on each side is exactly the same, as we shall see. Analogous remarks will hold for the even $r$ case.

We start with the special, though pivotal, case where we have linear orders of size $2g(r-1) + 1$ and $2g(r-1)$, as in Figure \ref{fig:lower_bound_odd_case}.
\begin{figure} [ht]
\centerline{\scalebox{0.38}{\includegraphics{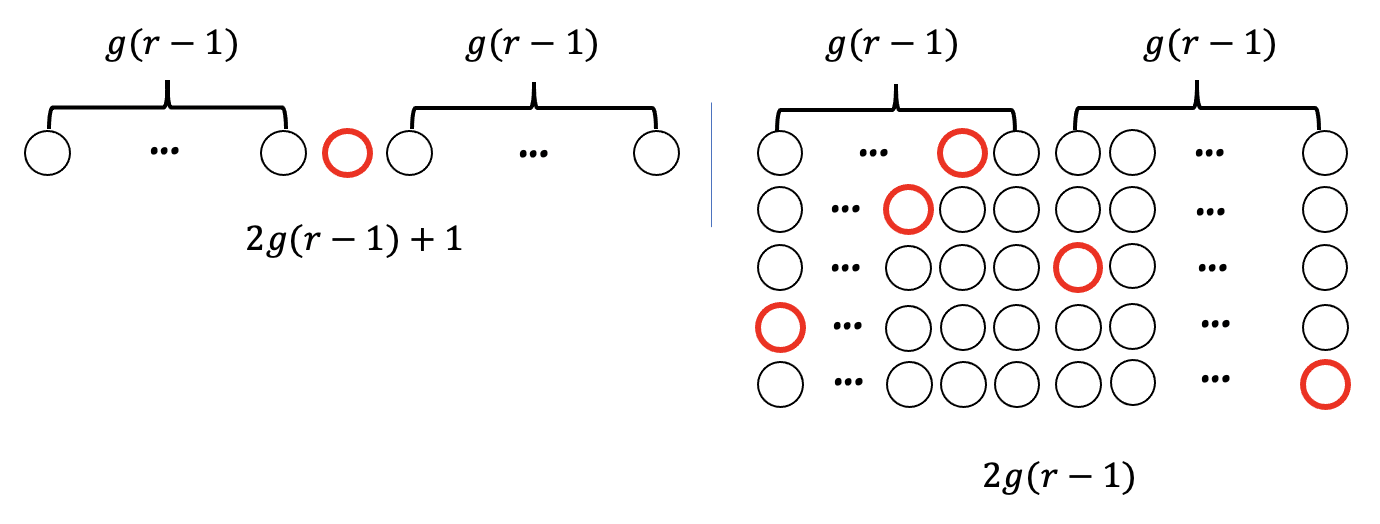}}}
\caption{The odd $r$ case: linear orders of sizes $2g(r-1) + 1$ and $2g(r-1)$. Spoiler plays first on $B$ (the left hand side). First round moves are indicated in red. Duplicator responds by playing every possible move. Five exemplary boards with different moves played on each are shown.}
\label{fig:lower_bound_odd_case}
\end{figure}
We describe a winning strategy for Spoiler. Spoiler begins by playing the middle element on $B$, i.e. $B(g(r-1) + 1)$. Duplicator responds by playing every possible move leaving short sides that all have fewer than $g(r-1)$ elements. Spoiler next makes his 2nd round move on each board of $L$ on the short sides, while playing on top of any end moves (moves with empty short sides), such as those shown on the bottom two boards in Figure \ref{fig:lower_bound_odd_case}. With the exception of the two end moves, which we will handle independently, Spoiler is going to play all subsequent moves entirely on the short sides of the boards of $L$ and the corresponding sides of the boards of $B$ -- in other words if the short side of $L$ is on the left, he will play on the left side of $B$, and vice versa -- in accordance with the known Spoiler-winning strategy on boards of size $g(r-1)$ or greater \vs boards of size less than $g(r-1)$. 

In order to see how he is able to do this, we need to use a strong inductive assumption, namely that in games of $r$ rounds, when $r$ is even, and there are linear orders of sizes $g(r)$ or greater on one side, and less than $g(r)$ on the other side, Spoiler can then always win by playing first on the $L$ side. The game-based proof of Lemma \ref{lemma:g_underline_of_4} demonstrated such a strategy for the case $r=4$, and we will have to keep this commitment when we cover additional even $r$ cases next. For now, though, we are assuming that $r$ is odd, so $r-1$ is even and we have such a strategy. 

Duplicator will respond with the oblivious strategy, making many copies of $B$ and playing all possible moves. From this point forward, the boards on both sides that have had their 2nd moves played to the right of the 1st moves conceptually constitute one game, and the boards on both sides that have had their 2nd moves played to the left of the 1st moves conceptually constitute a second game. The conceptual game played entirely on the left side of the boards can be won by Spoiler as well as the conceptual game played on the right, both times using the strong induction hypotheses. Note that since both games have the same larger size boards, Spoiler can win by picking the same side to play on ($L$ or $B$) on all boards in each of the conceptual games, in each subsequent round, and hence the two conceptual games can be played round-by-round in tandem. In so doing, all partial isomorphisms are broken and so Spoiler wins the combined game. 

However, we have not yet described how to take care of the case where Duplicator played end moves on $L$ in the 1st round and Spoiler reciprocated by playing on top of these end moves. In order to maintain an isomorphism with either of these boards Duplicator will have to play on top of the 1st move on $B$ -- which will break any potential isomorphism with any other $L$ boards. Now we use a key observation from the Spoiler-winning strategy in the 10+ \vs 9- game that established $g(4) \geq 10$ (see the proof of Lemma \ref{lemma:g_underline_of_4}): in the next-to-last round Spoiler played on $L$, and in the last round he played on $B$. Note that the $r$-round Spoiler strategy recursively uses an $(r-1)$-round Spoiler strategy, eventually using the $4$-round Spoiler strategy since the base case of this lemma is $r=5$, which is defined in terms of $r=4$.  Thus, for boards that remain isomorphic long enough, Spoiler will play the last two rounds consecutively on $L$ and then $B$.  
    
In the next-to-last round, in which Spoiler plays on $L$, Spoiler will select any element to the right of $L(1)$ on the boards where $L(1)$ was played as a first move, and Spoiler will select any element to the left of $L(2g(r-1))$ on the boards where $L(2g(r-1))$ was played as a 1st move. As a result, to maintain partial isomorphisms with the boards in which, respectively, $L(1)$ and $L(2g(r-1))$ were played, Duplicator will have to play so that the $B$ boards that remain isomorphic with the board in which $L(1)$ was played are not isomorphic to the boards in which $L(2g(r-1))$ was played, and vice versa. Thus, in the final round, Spoiler will play to the right of the middle element on the $B$ boards that maintained a partial isomorphism with the boards where $L(2g(r-1))$ was played first, thus killing all surviving partial isomorphisms, and will play to the left of the middle element on the $B$ boards that maintained a partial isomorphism with boards where $L(1)$ was played first, killing all of its remaining partial isomorphisms.

For the more general case where $|B| > 2 g(r-1) + 1$, Spoiler simply picks any element having at least $g(r-1)$ elements on each side of his 1st move and the play proceeds via the same induction. Analogously, if $|L| < 2g(r-1)$ it still follows that Duplicator's 1st round moves leave a short side of size less than $g(r-1)$, and hence the argument remains the same with respect to such smaller size $L$.

In the case where we start with multiple boards $\{B_i\}, \{L_j\}$ of sizes $|B_i| \geq 2 g(r-1) + 1$ and $|L_j| \leq 2 g(r-1)$, Spoiler starts by playing on each board on the $\{B_i\}$ side exactly as if there were just that one board there, and responds in turn, with the same alternating strategy as we have outlined in the single board case. Instead of making copies of just one board on her first round moves, Duplicator will make copies of all boards and make every possible move. Spoiler will then respond as before, playing on the short sides (or on top of end elements if necessary), and the argument recurses as in the case where we started with single boards.

Next let us tackle the case where $r$ is even. We begin as usual with the base case of boards of sizes $2g(r-1)$ and $2g(r-1) - 1$. See Figure \ref{fig:lower_bound_even_case}.

\begin{figure} [ht]
\centerline{\scalebox{0.38}{\includegraphics{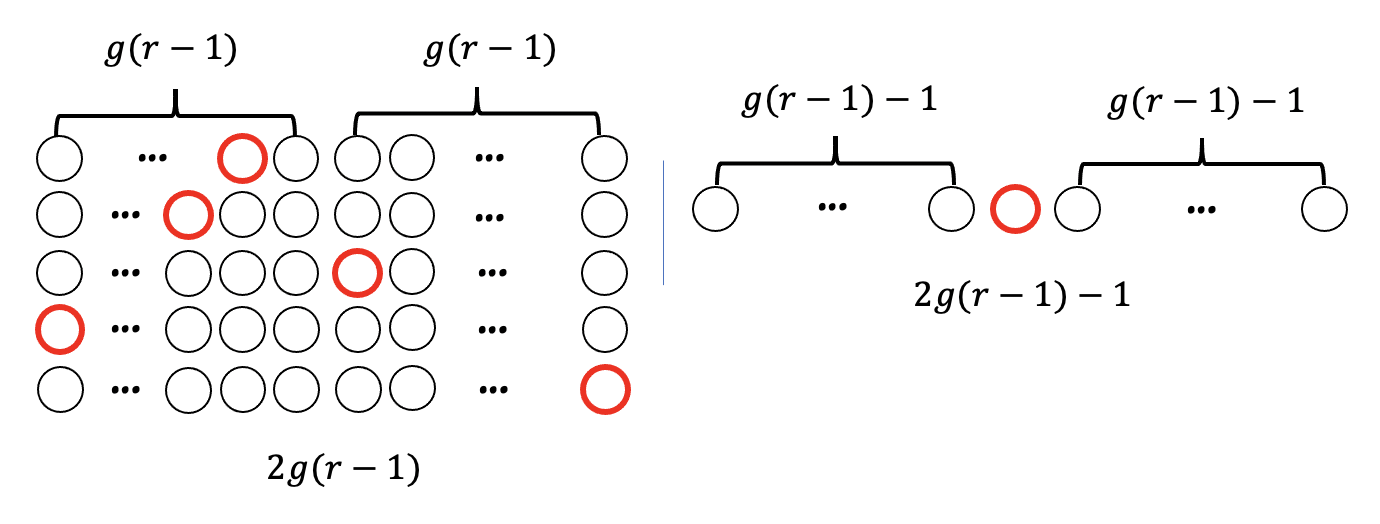}}}
\caption{The even $r$ case: linear orders of sizes $2g(r-1)$ and $2g(r-1) - 1$. Spoiler plays first on $L$ (the right hand side). First round moves are indicated in red. Duplicator responds by playing every possible move. Five exemplary boards with different moves played on each are shown.}
\label{fig:lower_bound_even_case}
\end{figure}

As Spoiler, we must keep our inductive commitment to play first on $L$ and so play the middle element, as indicated in the figure. Now any move by Duplicator on $B$ leaves a long side of size at least $g(r-1)$ versus the same side (left or right of the 1st move) on $L$, of size $g(r-1)-1$. Thus we adapt the argument from the odd $r$ case to this case, but where we play on the long side now rather than  on the short side. The strong induction hypothesis we need this time is that smaller odd $r$ cases can be won by playing the 1st round from $B$ -- but we know this to be the case for $r=5$, the very first case covered by this theorem, and all subsequent cases of odd $r$ by how we argued the odd $r$ case.  The case of boards of sizes greater than $2g(r-1)$ and less than $2g(r-1) - 1$ is handled analogously, and the same remarks about multiple boards versus multiple boards to start that we made for the odd $r$ case apply here as well. The theorem therefore follows.
\end{proof}

For $r > 4$, the sentences associated with the strategies described in the proof of the above theorem say, respectively, for $r$ even, that for every $x$ there is a linear order of size at least $g(r-1)$ either to the left or right of $x$, and for $r$ odd, that there is an element $x$ with a linear order of size at least $g(r-1)$ both to the left and right of $x$.

We give a sketch of how this works for $r = 5$ and $r = 6$.  To get started we rewrite the previous expression for $\Phi_4$, given by (\ref{cond1}) -- (\ref{cond6}), replacing the variables $x,y,z,w$, respectively with $x_2, x_3, x_4, x_5$:
\begin{small}
\begin{align}
	\Phi_4 = \forall x_2 \exists x_3 \forall x_4 \exists x_5(& \notag \\
		&x_2 < x_4 < x_3 \rightarrow (x_5 \neq x_4 \land x_2 < x_5 < x_3) ~~~\land \label{cond1a} \\
		&x_2 < x_3 < x_4 \rightarrow (x_5 \neq x_4 \land x_2 < x_3 < x_5) ~~~\land \label{cond2a}\\
		&x_3 < x_4 < x_2 \rightarrow (x_5 \neq x_4 \land x_3 < x_5 < x_2)~~~ \land\label{cond3a} \\
		&x_4 < x_3 < x_2 \rightarrow (x_5 \neq x_4 \land x_5 < x_3 < x_2)~~~\land \label{cond4a}\\
		&x_4 = x_2 \rightarrow (x_2 < x_5 < x_3 \vee x_3 < x_5 < x_2)~~~\land \label{cond5a}\\
		&x_4 = x_3 \rightarrow (x_2 < x_3 < x_5 \vee x_5 < x_3 < x_2)). \label{cond6a}
\end{align}
\end{small}
With this translation of variable names, the sentence says that ``for every $x_2$ there is a $x_3$ with two or more elements on each side of $x_3$, both of which are on the same side of $x_2$ as $x_3$''. More specifically, as noted before, the first four implications (\ref{cond1a})--(\ref{cond4a}) say that for every $x_2$ there is an $x_3$ such that for any $x_4 \neq x_2,x_3$, with $x_4$ on the same side of $x_2$ as $x_3$, there is an additional element besides $x_4$ on the same side of $x_3$ and also on the same side of $x_2$ as $x_3$. The last two implications (\ref{cond5a})--(\ref{cond6a}) insure that there are elements (i) between $x_2$ and $x_3$ and (ii) less than $x_3$ if $x_3 < x_2$, and greater than $x_3$ if $x_3 > x_2$.

To this we now add that there exists an element $x_1$ such that there is a linear order of size 10 both to its right and its left, in other words there exists an element $x_1$ such that the above sentence is true for elements less than $x_1$ and for elements greater than $x_1$, specifically:
\begin{small}
\begin{align*}
	\Phi_5 = \exists x_1 \forall x_2 \exists x_3 \forall x_4 \exists x_5\Big(& \notag \\
	x_1 < x_2 \rightarrow~~~~~&\Big( x_1 < x_2 < x_4 < x_3 \rightarrow (x_5 \neq x_4 \land x_1 < x_2 < x_5 < x_3) ~~~\land \label{cond1b} \\
		&x_1 < x_2 < x_3 < x_4 \rightarrow (x_5 \neq x_4 \land x_1 < x_2 < x_3 < x_5) ~~~\land \\ 
		&x_1 < x_3 < x_4 < x_2 \rightarrow (x_5 \neq x_4 \land x_1 < x_3 < x_5 < x_2)~~~ \land \\ 
		&x_1 < x_4 < x_3 < x_2 \rightarrow (x_5 \neq x_4 \land x_1 < x_5 < x_3 < x_2)~~~\land \\ 
		&x_4 = x_2 \rightarrow (x_1 < x_2 < x_5 < x_3 \vee x_1 < x_3 < x_5 < x_2)~~~\land \\ 
		&x_4 = x_3 \rightarrow (x_1 < x_2 < x_3 < x_5 \vee x_1 < x_5 < x_3 < x_2)\Big)~~\land \\ 
	x_2 < x_1 \rightarrow~~~~~&\Big( x_2 < x_4 < x_3 < x_1 \rightarrow (x_5 \neq x_4 \land x_2 < x_5 < x_3 < x_1) ~~~\land \\ 
		&x_2 < x_3 < x_4 < x_1 \rightarrow (x_5 \neq x_4 \land x_2 < x_3 < x_5 < x_1) ~~~\land \\ 
		&x_3 < x_4 < x_2 < x_1 \rightarrow (x_5 \neq x_4 \land x_3 < x_5 < x_2 < x_1)~~~ \land \\
		&x_4 < x_3 < x_2 < x_1 \rightarrow (x_5 \neq x_4 \land x_5 < x_3 < x_2 < x_1)~~~\land \\ 
		&x_4 = x_2 \rightarrow (x_2 < x_5 < x_3 < x_1 \vee x_3 < x_5 < x_2 < x_1)~~~\land \\ 
		&x_4 = x_3 \rightarrow (x_2 < x_3 < x_5 < x_1 \vee x_5 < x_3 < x_2 < x_1)\Big)~~\land \\ 
	x_1 = x_2 \rightarrow~~~~~&(x_3 < x_1 \land x_1 < x_5)\Big). 
\end{align*}
\end{small}
Then $\Phi_5$ distinguishes linear orders of size $21$ and above from those of size less than $21$. 

Next, to form the expression for $\Phi_6$, which distinguishes linear orders of size $42$ and above from those of size less than $42$, we say that for every element $x_0$ there is either a linear order of size $21$ to the left or right of $x_0$. The construction is analogous:
\begin{footnotesize}
\begin{align*}
	\Phi_6 = \forall x_0 \exists x_1& \forall x_2 \exists x_3 \forall x_4 \exists x_5\Bigg(& \notag \\
	x_0 < x_1 \rightarrow \Bigg(&x_0 < x_1 < x_2 \rightarrow&~~\Big( x_0 < x_1 < x_2 < x_4 < x_3 \rightarrow (x_5 \neq x_4 \land x_0 < x_1 < x_2 < x_5 < x_3) ~~~\land  \\
		&&x_0 < x_1 < x_2 < x_3 < x_4 \rightarrow (x_5 \neq x_4 \land x_0 < x_1 < x_2 < x_3 < x_5) ~~~\land \\
		&&x_0 < x_1 < x_3 < x_4 < x_2 \rightarrow (x_5 \neq x_4 \land x_0 < x_1 < x_3 < x_5 < x_2)~~~ \land \\
		&&x_0 < x_1 < x_4 < x_3 < x_2 \rightarrow (x_5 \neq x_4 \land x_0 < x_1 < x_5 < x_3 < x_2)~~~\land \\
		&&x_4 = x_2 \rightarrow (x_0 < x_1 < x_2 < x_5 < x_3 \vee x_0 < x_1 < x_3 < x_5 < x_2)~~~\land \\
		&&x_4 = x_3 \rightarrow (x_0 < x_1 < x_2 < x_3 < x_5 \vee x_0 < x_1 < x_5 < x_3 < x_2)\Big)~~\land \displaybreak \\ 
	&x_0 < x_2 < x_1 \rightarrow&~~\Big( x_0 < x_2 < x_4 < x_3 < x_1 \rightarrow (x_5 \neq x_4 \land x_0 < x_2 < x_5 < x_3 < x_1) ~~~\land  \\
		&&x_0 < x_2 < x_3 < x_4 < x_1 \rightarrow (x_5 \neq x_4 \land x_0 < x_2 < x_3 < x_5 < x_1) ~~~\land \\
		&&x_0 < x_3 < x_4 < x_2 < x_1 \rightarrow (x_5 \neq x_4 \land x_0 < x_3 < x_5 < x_2 < x_1)~~~ \land \\
		&&x_0 < x_4 < x_3 < x_2 < x_1 \rightarrow (x_5 \neq x_4 \land x_0 < x_5 < x_3 < x_2 < x_1)~~~\land \\
		&&x_4 = x_2 \rightarrow (x_0 < x_2 < x_5 < x_3 < x_1 \vee x_0 < x_3 < x_5 < x_2 < x_1)~~~\land \\
		&&x_4 = x_3 \rightarrow (x_0 < x_2 < x_3 < x_5 < x_1 \vee x_0 < x_5 < x_3 < x_2 < x_1)\Big)~~\land \\
	&(x_0 < x_1 \land x_1 = x_2) \rightarrow&~~(x_0 < x_3 < x_1 \land x_0 < x_1 < x_5)\Bigg) \\
	& \land & \\
		x_1 < x_0 \rightarrow \Bigg(&x_1 < x_2 < x_0 \rightarrow&~~\Big( x_1 < x_2 < x_4 < x_3 < x_0 \rightarrow (x_5 \neq x_4 \land x_1 < x_2 < x_5 < x_3 < x_0) ~~~\land \\
		&&x_1 < x_2 < x_3 < x_4 < x_0 \rightarrow (x_5 \neq x_4 \land x_1 < x_2 < x_3 < x_5 < x_0) ~~~\land \\
		&&x_1 < x_3 < x_4 < x_2 < x_0 \rightarrow (x_5 \neq x_4 \land x_1 < x_3 < x_5 < x_2 < x_0)~~~ \land \\
		&&x_1 < x_4 < x_3 < x_2 < x_0 \rightarrow (x_5 \neq x_4 \land x_1 < x_5 < x_3 < x_2 < x_0)~~~\land \\
		&&x_4 = x_2 \rightarrow (x_1 < x_2 < x_5 < x_3 < x_0 \vee x_1 < x_3 < x_5 < x_2 < x_0)~~~\land \\
		&&x_4 = x_3 \rightarrow (x_1 < x_2 < x_3 < x_5 < x_0 \vee x_1 < x_5 < x_3 < x_2 < x_0)\Big)~~\land \\
	&x_2 < x_1 < x_0 \rightarrow&~~\Big( x_2 < x_4 < x_3 < x_1 < x_0 \rightarrow (x_5 \neq x_4 x_2 < x_5 < x_3 < x_1 < x_0) ~~~\land  \\
		&&x_2 < x_3 < x_4 < x_1 < x_0 \rightarrow (x_5 \neq x_4 \land x_2 < x_3 < x_5 < x_1 < x_0) ~~~\land \\
		&&x_3 < x_4 < x_2 < x_1 < x_0 \rightarrow (x_5 \neq x_4 \land x_3 < x_5 < x_2 < x_1 < x_0)~~~ \land \\
		&&x_4 < x_3 < x_2 < x_1 < x_0 \rightarrow (x_5 \neq x_4 \land x_5 < x_3 < x_2 < x_1 < x_0)~~~\land \\
		&&x_4 = x_2 \rightarrow x_0 < (x_2 < x_5 < x_3 < x_1 \vee x_3 < x_5 < x_2 < x_1 < x_0)~~~\land \\
		&&x_4 = x_3 \rightarrow x_0 < (x_2 < x_3 < x_5 < x_1 \vee x_5 < x_3 < x_2 < x_1 < x_0)\Big)~~\land \\
	&(x_1 = x_2 \land x_1 < x_0) \rightarrow~~&(x_3 < x_1 < x_0 \land x_1 < x_5 < x_0)\Bigg)  \\
	& \land & \\
	x_0 = x_1 \rightarrow &~x_3 \neq x_0\Bigg).&
\end{align*}
\end{footnotesize}

It is worth noting that we could have begun this process at $g(4)$. Starting with the expression for $\Phi_3$, given by (\ref{g3_forall}), establishing  $g(3) \geq 4$ with prenex signature $\forall\exists\exists$ we could have ``relativized'' $\Phi_3$  to form an expression with prenex signature $\exists\forall\exists\exists$ saying that there exists an element with a linear order of size at least $4$ both to the left and right. This would have established that $g(4) \geq 9$. From there we could have obtained an expression with prenex signature $\forall\exists\forall\exists\exists$ establishing that $g(5)  \geq 18$, and so on, at each juncture obtaining slightly worse bounds than we have already obtained.  The magic is that a stronger lower bound for $g(4)$ can be expressed with the sentence given by (\ref{cond1}) -- (\ref{cond6}) [alternatively, \ref{cond1a}) -- (\ref{cond6a})].

\section{Games with Atoms: A different Approach to Obtain Upper Bounds on $g(r)$}\label{sec:atoms}

Let us define a new type of game. These games are similar to the \ms games but with a twist. They will make it harder for Duplicator to win any particular game of $r$ rounds. Since Duplicator-winning strategies for a game of $r$ rounds provide upper bounds on $g(r)$, we will obtain upper bounds that are potentially weaker, but will later match the lower bounds for $r \geq 4$, proving the upper bounds in the range $r \geq 4$ to be tight. The reason for considering these games is that they will allow us to recurse, getting around the issue we got stuck on in Section~\ref{sec:interlude}.

In our new game of $r$ rounds, we again have sets $\mathcal{A}$ and $\mathcal{B}$ of structures. However, each board on each of the sides, in addition to containing a structure $S$, contains a collection of unrelated, but labeled elements, $\{a_1,...,a_s\}$, where $s$ can be as large as Spoiler wants.
We shall refer to the collection of unrelated labeled elements as \textit{atoms}. 
By a slight abuse of notation, we will treat atoms as if they are constant symbols, so that they can appear in sentences. However, these symbols will not appear in the sentences 
guaranteed by Theorems~\ref{thm:main1} and Theorem~\ref{thm:main1a}.
These atoms are both unrelated to elements of the structure and unrelated to each other.
On their respective turns, Spoiler and Duplicator play, as in the \ms games, on all boards on their chosen side for that turn, picking an element from that board, choosing either an element from the structure, or from the set of atoms. If atom $a_k$ is selected by Spoiler on a given board, in a given round $j$, the only way Duplicator can maintain a partial isomorphism between that board and an opposite board is to select $a_k$ in round $j$ on the opposite board. 
When Duplicator makes copies of boards, the atoms are copied as well as the structures labeled with the moves made thus far.

In our figures, rather than showing all the atoms, and distinguishing those that are selected, we show only the atoms that have thus far been selected for a given board adjacent to the structures, with appropriate labeling.
Let us refer to these new games as \textit{\ms games with atoms}. 

\begin{observ}\label{obs:eq}
\ms games with atoms are equivalent to \ms games with structures that are a union of the prior structures and the set of atoms.
Hence, by Theorem~\ref{thm:main1}, we have that Spoiler wins $r$-round \ms games with atoms on two sets of structures iff there exists a sentence $\Phi$ with up to $r$ quantifiers that distinguishes the two sets of structures. The sentence can contain constants and utilize an additional unary relation $A(x)$, which is true iff $x$ is an atom. 
\end{observ}

Since a Spoiler winning strategy for a \ms game is still a winning strategy for the analogous game with atoms (where atoms are never selected by Spoiler) we have:
\begin{lem} \label{lemma:atoms}
If Spoiler can win an $r$-round \ms game on a given pair of sets of structures then he can also win an $r$-round \ms game with atoms on the same pair of sets of structures.	
\end{lem}

\begin{cor} \label{cor:atoms}
If Duplicator can win an $r$-round \ms game with atoms on a given pair of sets of structures, then she can also win an $r$-round \ms game on the same pair of sets of structures.
\end{cor}

To understand the additional power provided to Spoiler by the atoms, let us revisit how Duplicator survives in one critical juncture of the $3$-round,
$5$ versus $4$ \ms game from Lemma~\ref{lemma:g_upper_of_3}.
Figure~\ref{fig:B5L4} 
\begin{figure}[ht]
\begin{center}
\includegraphics[width=3.4in]{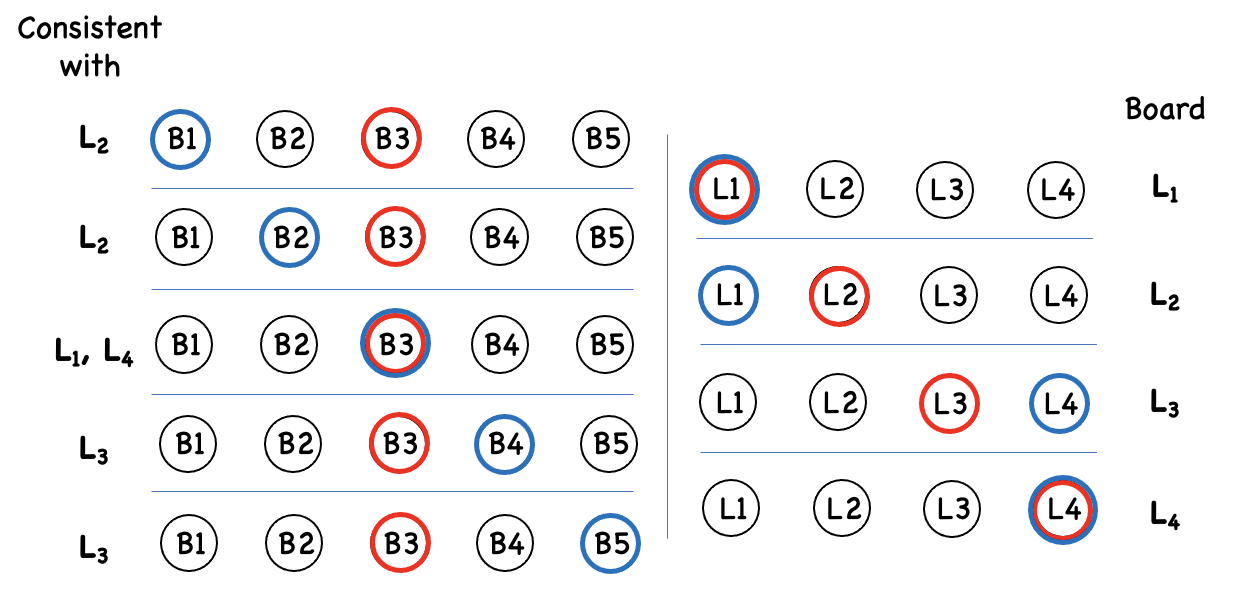}
\caption{MS game play on $\mathcal{B}$ versus $\mathcal{L}$. 1st round moves are in Red, 2nd round moves in blue.}\label{fig:B5L4}
\end{center}
\end{figure}
shows the boards after two rounds in the critical sequence. In the 1st round, Spoiler played $B(3)$ and Duplicator played all possible moves on $L$. We consider the case where Spoiler then replies with 2nd round moves on $L$ as indicated in blue in the figure, and in particular playing on top of $L(1)$ on the board labeled $L_1$ and on top of $L(4)$ on the board labeled $L_4$. Duplicator then replies with all possible responses on $B$ (again in blue). To the left of each $B$ board we indicate which $L$ boards the board has managed to keep an isomorphism with. 

The crux of the matter is that because the third $B$ board is consistent with both $L_1$ and $L_4$, Spoiler cannot break both isomorphisms in the one remaining move.

With atoms, however, investing a turn to play an atom  can usefully separate these two $L$ boards.  Figure~\ref{fig:B5L4Atoms} 
\begin{figure}[ht]
\begin{center}
\includegraphics[width=3.4in]{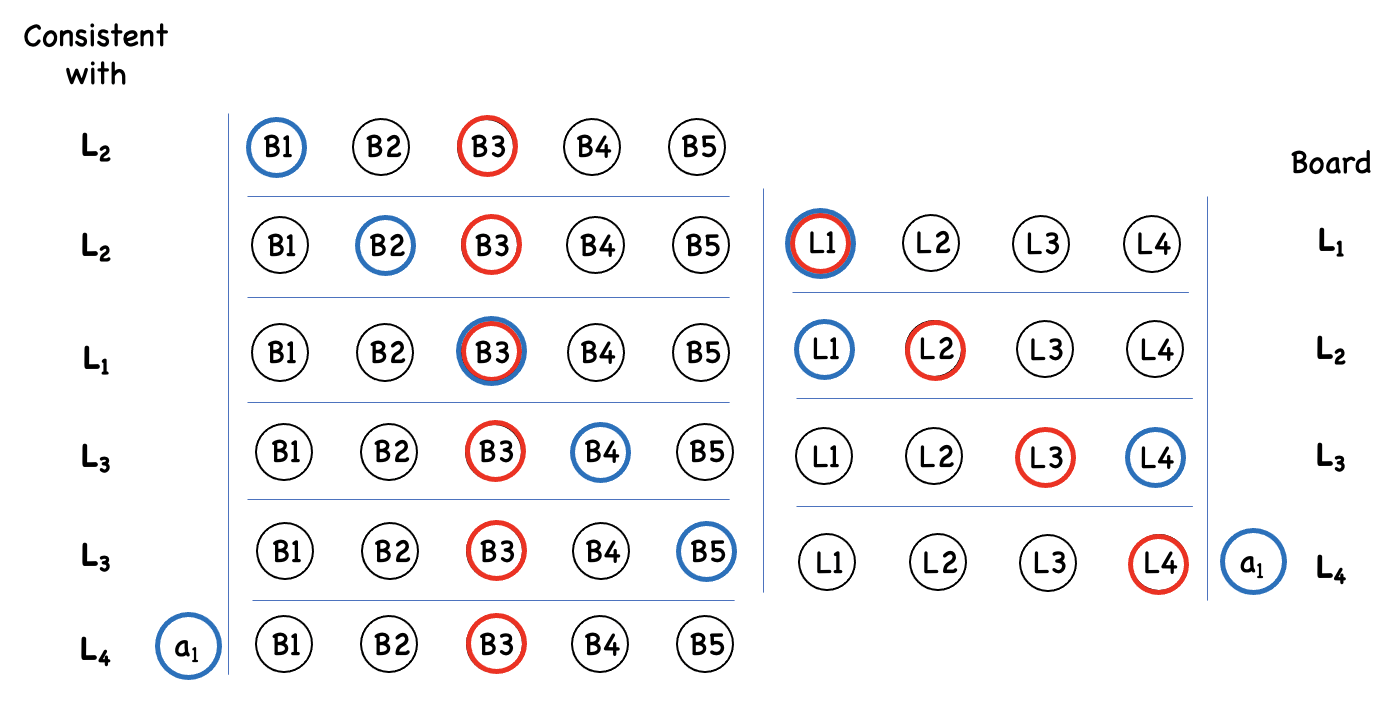}
\caption{MS game play with atoms.  
}\label{fig:B5L4Atoms}
\end{center}
\end{figure}
illustrates Spoiler changing his second move on board $L_4$ from $L(4)$ to $a_1$.  Duplicator now has an additional potentially viable 2nd round move, playing the newly played atom, $a_1$, as indicated in the additional board added on the left.  However, note in the left hand column of the figure how no board is any longer simultaneously partially isomorphic with boards $L_1$ and $L_4$. Spoiler can now break each of the isolated isomorphisms by playing on $B$, in order, from top to bottom, as follows: $B(2), B(1), B(2), B(5), B(4), B(4)$.

\medskip

The following follows from the discussion above and Lemma \ref{lemma:g_upper_of_3}.

\begin{prop} \label{lemma:atoms_vs_no_atoms_example}
While Duplicator can always win $3$-round \ms games on linear orders of sizes $5$ \vs $4$, Spoiler can always win $3$-round \ms games with atoms on linear orders of these sizes.
\end{prop}

\medskip

We are now going to establish Duplicator-winning strategies for \ms games with atoms, for various numbers of rounds. By Lemma \ref{lemma:atoms}, Duplicator has a harder time winning these games than standard \ms games and so these strategies will initially provide weaker upper bounds than the upper bounds obtained by providing Duplicator-winning strategies in \ms games. 

\begin{defi}
Let $g'(r)$ denote the largest $k$ such that Spoiler wins every $r$-round \ms game with atoms on two sets $\mathcal{B}$ and $\mathcal{L}$ of linear orders  where each $B \in \mathcal{B}$ is of size at least $k$ and each $L \in \mathcal{L}$ is of size less than $k$.
\end{defi} 

To see that $g'(r)$ is well-defined, note that if Duplicator can win an $r$-round E-F game on two linear orders $B, L$ then she can win an $r$-round \ms game, as well as an $r$-round \ms game with atoms, on sets $\mathcal{B},\mathcal{L}$ of structures with $B \in \mathcal{B}$ and $L \in \mathcal{L}$ -- this is the case because she can focus on the one pair of structures, $L$ and $B$, when playing the \ms game, and further, because Spoiler gains no advantage from playing atoms when play is constrained 
to just a pair
of structures. Hence $g'(r) \leq f(r)$. In our considerations of $g'$ we will exclusively be focusing on game-based approaches. We will, further, just be concerned with establishing upper bounds on $g'(r)$, and in so doing will always be taking $\mathcal{B}$ and $\mathcal{L}$ to be singletons, $\mathcal{B} = \{B\}, \mathcal{L} = \{L\}$.

The following is an immediate consequence of Lemma \ref{lemma:atoms}.
\begin{lem} \label{lemma:g'_doms_g}
For all values of $r$, we have  $g(r) \leq g'(r)$.
\end{lem}

The following lemma describes why it is possible to recursively prove upper bounds on $g'$ (and hence $g$) in \ms games with atoms and get around the issue described in Section~\ref{sec:interlude}. 

\begin{lem} \label{lemma:reduction}
\textbf{Reduction Lemma:} Suppose we have an $r$-round \ms game with atoms on linear orders of sizes $K$ and $K'$.  For some integer $h$, with $1 \leq h \leq \min(K, K')$, suppose there are boards on the $L$ and $B$ sides in which first moves of $L(h)$ and $B(h)$ have been played, or in which moves of $L(K-h+1)$ and $B(K'-h+1)$ have been played. Then Duplicator wins the $K$ \vs $K'$ game on those boards iff she wins the $(r-1)$-round \ms game with atoms on linear orders of sizes $K-h$ and $K'-h$.
\end{lem}

Figure \ref{fig:6_vs_5_reduction} illustrates the reduction of a $3$-round, $6$ versus $5$ \ms game with atoms, to a $2$-round, $3$ versus $2$ \ms game with atoms, per the conclusion.

\begin{figure} [ht]
\centerline{\scalebox{0.30}{\includegraphics{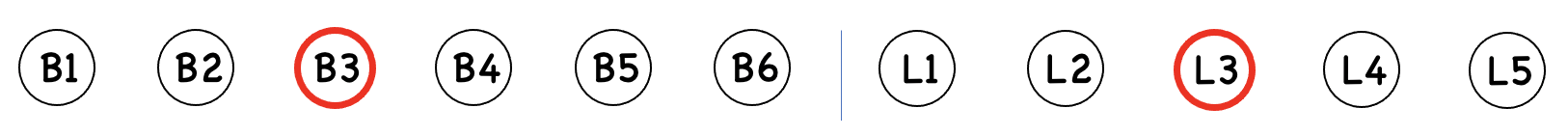}}}
\caption{An illustration of the reduction in Lemma \ref{lemma:reduction} where a $3$-round, $6$ versus $5$ \ms game with atoms gets reduced to a $2$-round, $3$ versus $2$ \ms game with atoms.}
\label{fig:6_vs_5_reduction}
\end{figure}

\begin{proof}
Without loss of generality assume the 1st round moves are $L(h)$ and $B(h)$. 
Duplicator now makes a pact with Spoiler, saying that the only way she will maintain an isomorphism with moves on the left hand side of $L(h)$/$B(h)$ is by mirroring (in other words by responding to $L(i)$ with $B(i)$ when $1 \leq i \leq h$, and vice versa). She tells Spoiler that if ever such a move is \textit{not} mirrored, Spoiler can count it as a break in the partial isomorphism between the two boards. By agreeing to these stricter isomorphism rules, Duplicator makes it harder for herself to maintain partial isomorphisms.  However, the effect is that we can remove the elements $B(1),...,B(h)$ and $L(1),...,L(h)$ and set aside the first $h$ atoms on each side, $a_1,...,a_h$, only to be played as follows: if Spoiler ever would have wanted to play $B(i)$ or $L(i)$, with $1  \leq i \leq h$, on a given board, he can instead play $a_i$ with the same effect. This reduces the game to an $(r-1)$-round \ms game with atoms on linear orders of size $K-h$ and $K'-h$ with $h$ additional atoms. Since in the definition of \ms games with atoms, Spoiler already has as many atoms as he wants, the additional $h$ atoms have no effect on the game.
We are left with just playing an $(r-1)$-round \ms game with atoms on linear orders of size $K-h$ and $K'-h$ and if Duplicator can win such a game, she can win the original game.

On the other hand, if Duplicator does \textit{not} have a winning strategy in the $K-h$ \vs $K'-h$ \ms game with atoms, then Spoiler has a winning strategy and can force play to be entirely on the side of the initial move with this many unplayed elements, and 
thus win.  The lemma follows.
\end{proof}

\begin{lem} \textbf{Laddering Up Lemma for \MS Games and \MS Games with Atoms:} \label{lemma:stepping-up}
Suppose Duplicator can win \ms games (respectively \ms games with atoms) on singleton sets of sizes $K, K+1$ for all $K \geq N$. Then Duplicator can also win \ms games (respectively \ms games with atoms) 
on singleton sets of sizes $K, K'$ whenever $K \geq N$ and $K' \geq N$.
\end{lem}

\begin{proof}
Suppose both $K \geq N$ and $K' \geq N$. Let l.o.($K$) denote the linear order of size $K$. Then we have that l.o.($K$) $\equiv_r$ l.o.($K+1$) $\equiv_r \cdot\cdot\cdot \equiv_r$ l.o.($K'$). By repeated application of Lemma \ref{lemma:prenex_equiv} the lemma follows for \ms games.

The same argument works for \ms games with atoms by replacing each linear order with the union of the linear order and the corresponding atoms, whereby \ms games with atoms reduce to \ms games, per Observation~\ref{obs:eq}. 
\end{proof}

\begin{lem} \label{lemma:g_prime_of_2}
Duplicator can win $2$-round MS games with atoms on linear orders of sizes $2$ or greater and hence $g'(2) \leq 2$.
\end{lem}

\begin{proof}
Immediate from Lemma \ref{lemma:g_upper_of_2} coupled with the observation that it never helps Spoiler to play an atom in the first or last round.
\end{proof}

\begin{lem} \label{lemma:g_prime_of_3}
Duplicator can win $3$-round \ms games with atoms on linear orders of sizes $5$ or greater and hence $g'(3) \leq 5$.
\end{lem}

\begin{proof}
Suppose we have linear orders of sizes $K,K+1$ with $K \geq 5$. Any Spoiler 1st round move leaves a short side of no more than $\lfloor\frac{K}{2}\rfloor$ unplayed elements on that side. Duplicator can then reply with a move leaving an identical short side. Without loss of generality assume this common short side is on the left. Then to the right of the played moves, each board has at least $K - \lfloor\frac{K}{2}\rfloor - 1 \geq 2$ unplayed elements, and by virtue of the fact that $g'(2) \leq 2$ (Lemma \ref{lemma:g_prime_of_2}), Lemma \ref{lemma:reduction} guarantees that Duplicator has a winning strategy. The Laddering Up Lemma \ref{lemma:stepping-up} then guarantees that Duplicator has a winning strategy for any boards of sizes $K \geq 5$ and $K' \geq 5$. The lemma follows.
\end{proof}

Although this section is concerned with establishing upper bounds on $g$, we shall actually need to establish precise values for $g'$ in order to get the upper bounds on $g$ to go through. The discussion we gave to show that Spoiler can win a $5$ \vs $4$ $3$-round \ms game with atoms (Proposition \ref{lemma:atoms_vs_no_atoms_example}) can easily be extended to show that Spoiler can win such a $3$-round game on linear orders of sizes $5$ or greater versus $4$ or smaller. Hence we have that $g'(3) \geq 5$, so that together with the prior lemma we have:

\begin{lem} \label{lemma:g_prime_of_3_hard} $g'(3) = 5$.
\end{lem}

Any Spoiler-winning strategy in an ordinary \ms game of $r$ rounds corresponds to a sentence that is valid for $B$ but not for $L$ with $r$ quantifiers. If Spoiler's strategy starts on $B$ then the corresponding sentence starts with $\exists$, while if Spoiler's strategy starts on $L$, the sentence starts with $\forall$. If we force Spoiler to play first on $L$, then we are giving Duplicator an advantage so that she may be able to win games that she would not be able to win otherwise.

\begin{defi} \label{def:g_forall}
Let $g'_\forall(r)$ denote the smallest $k$ such that Duplicator can win \ms games with atoms on a pair of linear orders, each of size $k$ or greater, if Spoiler is constrained to play his first move on $L$.
\end{defi}

\begin{lem} \label{lemma:for_all}
If Spoiler is constrained to play his first move on $L$,  Duplicator can win $r$-round \ms games with atoms on linear orders of sizes $2g'(r-1)$ or greater. Hence, 
$g'_\forall(r) \leq 2g'(r-1)$. 
\end{lem}

\begin{proof} 
Consider an $r$-round, $2g'(r-1) + 1$ \vs $2g'(r-1)$ \ms game with atoms. Refer to Figure \ref{fig:forall_lemma}.
\begin{figure} [ht]
\centerline{\scalebox{0.38}{\includegraphics{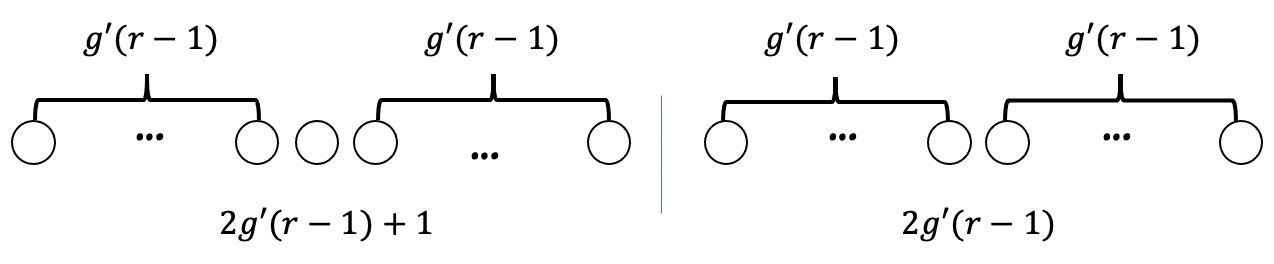}}}
\caption{The $r$-round, $2g'(r-1) + 1$ \vs $2g'(r-1)$ \ms game with atoms. Only one board on each side are shown.}
\label{fig:forall_lemma}
\end{figure}
If Spoiler is constrained to play on the $L$ side, his play will necessarily leave a short side of size at most $g'(r-1)-1$, which can be matched by a move that leaves the same short side on $B$, and with long sides that are each of size at least $g'(r-1)$. Hence the game is winnable by Duplicator via the Reduction Lemma (Lemma \ref{lemma:reduction}). For $r$ round games on boards of sizes $K+1$ \vs $K$ with $K  > 2g'(r-1)$ again the short sides can be matched up, leaving long sides of sizes still at least $g'(r-1)$. The lemma follows by the Laddering Up Lemma (Lemma \ref{lemma:stepping-up}).
\end{proof}

\begin{thm} \label{thm:main_upper}
For $r \geq 2$, 
\begin{equation} \label{eqn:main_recursion}
	g'(r) = \begin{cases} 2g'(r-1)~~~~~~\textrm{if $r$ is even,}\\ 2g'(r-1) + 1\textrm{ if $r$ is odd.} \end{cases}
\end{equation}
Moreover, Duplicator can win $r$-round \ms games with atoms on linear orders of sizes $2g'(r-1)$ or greater if $r$ is even, and on linear orders of sizes $2g'(r-1) + 1$ or greater if $r$ is odd.
\end{thm}

\begin{proof}
We establish the equality asserted in the theorem via first establishing that the $\geq$ inequality holds, and then establishing that the $\leq$ inequality holds. Since for all $r \geq 1$ we have $g'(r) \geq g(r)$ (lemma \ref{lemma:g'_doms_g}), the $\geq$ portion of the theorem follows from the lemmas establishing that $g(1) = 1, g(2) \geq 2, g'(3) = 5$ and $g(4) \geq 10$, together with Theorem \ref{thm:main_lower}. 

Now let us establish that the $\leq$ inequality holds. It is trivial to verify that $g'(1) = 1$, while $g'(2) \leq 2$ is Lemma \ref{lemma:g_prime_of_2} and $g'(3) = 5$ is Lemma \ref{lemma:g_prime_of_3_hard}. Hence, we have already established the $\leq$ part of the theorem for $r=2$ and $r = 3$. For larger values of $r$ we proceed inductively, assuming the truth of the theorem for values up to $r-1$ and proving it for the value $r$.  The inductive step for the case of odd $r$ is easy, so let's dispose of that case first -- it is essentially the same argument we gave to establish $g'(3) \leq 5$ in Lemma \ref{lemma:g_prime_of_3}. Suppose we have linear orders of sizes $K, K+1$ with $K \geq 2g'(r-1) + 1$. Any Spoiler 1st round move leaves a short side of no more than $\lfloor\frac{K}{2}\rfloor$ unplayed elements on that side. Duplicator can then reply with a move leaving an identical short side. Without loss of generality assume this common short side is on the left. Then, to the right of the played moves, each board has at least $K - \lfloor\frac{K}{2}\rfloor - 1$ unplayed elements. But, using the inductive assumption, $K - \lfloor\frac{K}{2}\rfloor - 1 \geq K - \frac{K}{2} - 1 \geq (g'(r-1) + \frac{1}{2}) - 1$; hence $K - \lfloor\frac{K}{2}\rfloor - 1 \geq g'(r-1) - \frac{1}{2}$. Since both $K - \lfloor\frac{K}{2}\rfloor - 1$ and $g'(r-1)$ are integers, it follows that $K - \lfloor\frac{K}{2}\rfloor - 1 \geq g'(r-1)$. Thus, each board has at least $g'(r-1)$
unplayed elements on their long sides. The Reduction Lemma (Lemma \ref{lemma:reduction}) in conjunction with the induction hypothesis therefore guarantees that Duplicator has a winning strategy. The Laddering Up Lemma (Lemma \ref{lemma:stepping-up}) then guarantees that Duplicator has a winning strategy for any boards of sizes $K,K' \geq 2g'(r-1) + 1$.

Next consider the case of even $r$. Suppose first that we have linear orders of sizes $2g'(r-1)$ and $2g'(r-1) + 1$. Let us first dispose of any first move by Spoiler other than $B(g'(r-1) + 1)$, the middle element on the $B$ side. Any move other than this one on the $B$ side can be responded to with a move on $L$ that matches the short side while leaving at least $g'(r-1)$ unplayed elements on both long sides, and therefore, again by the Reduction Lemma and induction, yielding a winning strategy for Duplicator. On the other hand, any Spoiler 1st round move on $L$ is analogously met by matching the short side with a move on $B$, transposing to the just-prior analysis. Thus we may assume that Spoiler plays the element $B(g'(r-1) + 1)$. In response, Duplicator uses two boards and plays $L(g'(r-1))$ on one of the boards and $L(g'(r-1) + 1)$ on the other one.  See Figure \ref{fig:split_board}.
\begin{figure} [ht]
\centerline{\scalebox{0.38}{\includegraphics{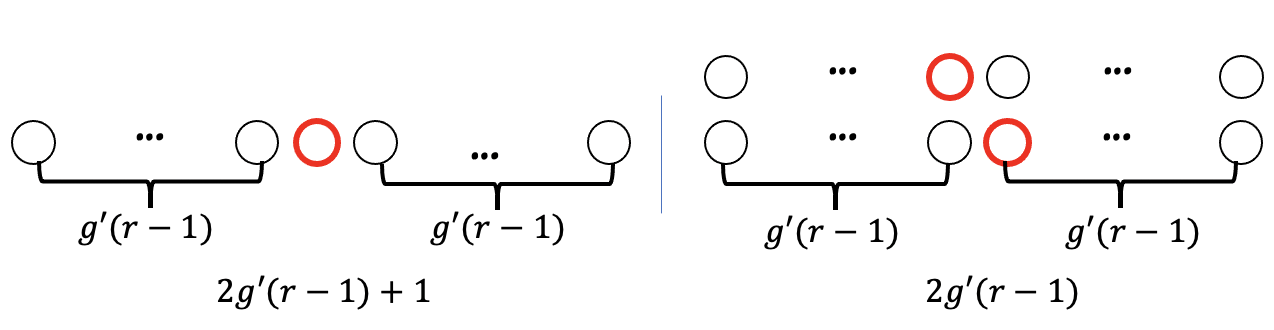}}}
\caption{A  $2g'(r-1) + 1$ \vs $2g'(r-1)$ game where Spoiler plays first on $B(g'(r-1) + 1)$ and Duplicator responds by playing $L(g'(r-1))$ on one board and $L(g'(r-1) + 1)$ on another.}
\label{fig:split_board}
\end{figure}

Consider the possible Spoiler 2nd round responses. Suppose Spoiler plays on $B$. A play on one of the left-hand unplayed $g'(r-1)$ elements, i.e., on some $B(i)$ such that $1 \leq i \leq g'(r-1)$ will be met with a play of $L(i)$ on the 2nd $L$ board. As we argued in the proof of the Reduction Lemma, we can now regard the remaining $g'(r-1) - 1$ unplayed elements to the left of the 1st round selections on the bottom $L$ board and the $B$ boards as additional atoms and just consider the game on the right side of these boards, which is now an $r-2$ round game on boards of sizes $g'(r-1)$ and $g'(r-1) - 1$. Inductively it is easy to see that $g'(r-1) - 1 > g'(r-2)$ and so such a sequence leads inductively to a Duplicator win. A 2nd round Spoiler play on $B$, on one of the right hand set of unplayed $g'(r-1)$ elements, is handled with a symmetrical argument. Further, playing an atom on $B$ is met by playing the same atom on both $L$ boards (in fact playing on just one board and ignoring the other is sufficient), and again leads to an inductive win for Duplicator by virtue of the fact that $g'(r-1) - 1 > g'(r-2)$. 

Thus we may assume that Spoiler makes his 2nd round moves on $L$. A play on the long side of either board is met with a symmetrical play by Duplicator on $B$, transposing to our prior analysis for when Spoiler played his 2nd move on $B$. Playing an atom on $L$ is met with the same atom being played on $B$, again with a transposition.  Thus we may suppose that Spoiler plays on the short side of both $L$ boards, and, in particular, plays on the left on the top $L$ board. Now the left hand (short side) of the top board is of size $g'(r-1) - 1$ and the left hand side of $B$ is of size $g'(r-1)$. Further, we are assuming that $r$ is even, and so $r-1$ is odd and hence by (\ref{eqn:main_recursion}) we have 
\begin{equation} \label{eqn:r-1_to_r-2}
    g'(r-1) = 2g'(r-2) + 1. 
\end{equation}
Since we've reduced the analysis to Spoiler next playing on the $L$ side of this $(r-1)$-round, $g'(r-1)$ \vs $g'(r-1)-1$  game, Lemma \ref{lemma:for_all} applies and says that $g'_\forall(r-1) \leq 2g'(r-2)$. Taken together with (\ref{eqn:r-1_to_r-2}), we have that $g'_\forall(r-1) \leq g'(r-1) - 1$.

Thus, we are left with boards of sizes $g'(r-1)$ and $g'(r-1)-1$, both of which are at least of the size of $g'_\forall(r-1)$. Hence, by the definition of $g'_\forall$ (Definition \ref{def:g_forall}), Duplicator has a winning strategy, over the remaining $r-1$ rounds, playing just on the left hand sides of these two boards and hence, by the Reduction Lemma, has a winning strategy playing on the entire board. Thus Duplicator has a winning strategy in the original $r$-round game for boards of sizes $2g'(r-1)$ and $2g'(r-1) + 1$.

For $K, K+1$ with $K > 2g'(r-1)$ the argument is easier since Duplicator can just mimic the short side play of any 1st round Spoiler play and immediately apply the Reduction Lemma. As usual, the argument is completed by applying the Laddering Up Lemma.
\end{proof}

We have therefore established the following upper bounds:

\begin{cor} \label{cor:main_upper}
We have
$g(2) \leq 2, g(3) \leq 4, g(4) \leq 10$, and for $r > 4$,
\begin{equation*} 
	g(r) \leq \begin{cases} 2g(r-1)~~~~~~\textrm{if $r$ is even,}\\ 2g(r-1) + 1\textrm{ if $r$ is odd.} \end{cases}
\end{equation*}
Moreover, Duplicator can win $r$-round \ms games on linear orders of sizes that are at least as large as these upper bound (right hand side) values in all the inequalities.
\end{cor}

\begin{proof}
The inequalities $g(2) \leq 2, g(3) \leq 4$ are Lemmas \ref{lemma:g_upper_of_2} and \ref{lemma:g_upper_of_3}. The chain of inequalities $g(4) \leq g'(4) \leq 2g'(3) \leq 10$, follows from  Lemma \ref{lemma:g'_doms_g}, Lemma \ref{lemma:g_prime_of_3} and Theorem \ref{thm:main_upper}. The inductively defined inequality for $r > 4$ follows by Lemma \ref{lemma:g'_doms_g} and Theorem \ref{thm:main_upper}. Finally, the upper bounds associated with all these lemmas and Theorem \ref{thm:main_upper} were established and stated by observing that Duplicator could win $r$ round games when the two linear orders considered were at least as large as the upper bounds. The same therefore follows for this corollary.
\end{proof}

\section{Tight Bounds on $g(r)$ and Final Theorems} \label{sec:final}
At last,
we pull together the identical upper and lower bounds we have obtained for $g(r)$, 
to yield Theorem \ref{thm:g}. Further, since the upper and lower bounds on $g(r)$ turned out to be tight, the very last paragraph at the end of Section \ref{sec:lower_bounds} immediately implies the following:

\begin{thm}
For each $r \geq 1$ there is a sentence with $r$ quantifiers distinguishing linear orders of size $g(r)$ or greater from linear orders of size less than $g(r)$. The prenex signatures of such sentences are as follows:
\begin{eqnarray*}
r = 1:& \exists \\
r = 2:& \exists\exists \\
r = 3:& \forall\exists\exists \\
r \geq 4, r~\textrm{even}:& \forall\exists\cdot\cdot\cdot\forall\exists \\
r \geq 4, r~\textrm{odd}:& \exists\forall\exists\cdot\cdot\cdot\forall\exists.
\end{eqnarray*}
\end{thm}

\noindent Here the $\cdot\cdot\cdot$ signifies a sequence of repeating quantifier alternations $\forall\exists$ of the length needed to give rise to $r$ quantifiers in total.

\begin{cor}
Let $\mathcal{A}$ be any collection of linear orders, each of size at least $g(r)$, and $\mathcal{B}$ any collection of linear orders, each of size less than $g(r)$. The Spoiler can win an $r$-round MS game  on $\mathcal{A}$ and $\mathcal{B}$.
\end{cor}

Although we have a completely specified $g(r)$, we have not completely answered the question of the minimum number of quantifiers needed to distinguish one linear order, $B$, from another, $L$, in one special case. Specifically, if $g(r-1) \leq |L| < |B| < g(r)$ for some value of $r$, we know there is no sentence with $r-1$ quantifiers that distinguishes $B$ from $L$, but nothing more. The following lemma closes this gap.

\begin{lem} \label{lemma:gap_closer} Given two linear orders $B, L$, with $|L| < |B| < g(r)$ for some $r > 0$ then there is a sentence with $r$ quantifiers that distinguishes $B$ from $L$.
\end{lem}

\begin{proof} 
If there is a sentence with $r-1$ quantifiers that distinguishes $B$ from $L$, there is obviously a sentence with $r$ quantifiers that does the same. Hence, we only need to verify the lemma for the case $g(r-1) \leq |L| < |B| < g(r)$. This is the range of values of $|L|$ and $|B|$ for which we have not resolved how many quantifiers suffice to separate $B$ from $L$.

We will demonstrate the conclusion of the lemma using a combination of explicit sentences and Spoiler winning strategies. It is not possible to have $|L| < |B| < g(1)$, so that case is handled. For $g(2)$ the only possibility is $|L| = 0, |B| = 1$, which is covered by $g(1) = 1$. Next, for $g(3)$, the cases covered by the lemma are those given in the table below.

\medskip

\begin{center}
\begin{tabular}{ |c|c|c|} 
\hline
$\bm{|L|}$ & $\bm{|B|}$ & \textbf{Why Spoiler Wins in $\bm{3}$ Rounds} \\
\hline
0 & 1 & $g(1) = 1$ \\
0 & 2 & $g(1) = 1$ \\
0 & 3 & $g(1) = 1$ \\
1 & 2 & $g(2) = 2$ \\
1 & 3 & $g(2) = 2$ \\
2 & 3 & $3$ rounds; Spoiler wins by playing 3 distinct elements in $B$ \\
\hline
\end{tabular}
\end{center}

\medskip

Next consider $g(4)$, which we shall consider as a special case.  
Everything else will follow via an induction argument based on $g(r)$ for $r > 4$. Since $g(4) = 10,$ we have $ |B| \leq 9$. The following expression with $4$ quantifiers distinguishes linear orders of size $9$ and above from those of size $8$ and below:

\begin{small}
\begin{align*}
	\Phi_{4, 9} = \forall x \exists y \forall z \exists w(& \notag \\
		&x < z < y \rightarrow (w \neq z \land x < w < y) ~~~\land  \\
		&x < y < z \rightarrow (w \neq z \land x < y < w) ~~~\land \\
		&y < z < x \rightarrow (w \neq z \land y < w < x) ~~~\land \\
		&z = x \rightarrow (x < w < y \vee y < w < x) ~~~\land \\
		&z = y \rightarrow (x < y < w \vee w < y < x)). 
\end{align*}
\end{small}

The sentence $\Phi_{4, 9}$ is constructed from the sentence $\Phi_{4}$ used in the proof of Lemma \ref{lemma:g_underline_of_4} by removing Condition~(\ref{cond4}) of $\Phi_4$. In words, $\Phi_{4, 9}$ states that ``for  every $x$ either there are 5 elements after it or 4 elements before it''. In more detail, ``for  every $x$ there is a $y$ such that if $y > x$ then there are two or more elements on each side of $y$, both of which are greater than $x$ and if $y < x$ then there are two or more elements between $y$ and $x$ and one element smaller than $y$''.

It is similarly easy to construct $\Phi_{4, 8}$, $\Phi_{4, 7}$, $\Phi_{4, 6}$ and $\Phi_{4, 5}$ just by removing more conditions from $\Phi_{4}$. We present these sentences below for completeness.
\begin{small}
\begin{align*}
	\Phi_{4, 8} = \forall x \exists y \forall z \exists w(& \notag \\
		&x < y < z \rightarrow (w \neq z \land x < y < w) ~~~\land \\
		&y < z < x \rightarrow (w \neq z \land y < w < x) ~~~\land \\
		&z = x \rightarrow (x < w < y \vee y < w < x) ~~~\land \\
		&z = y \rightarrow (x < y < w \vee w < y < x)). 
\end{align*}
\end{small}

\begin{small}
\begin{align*}
	\Phi_{4, 7} = \forall x \exists y \forall z \exists w(& \notag \\
		&x < y < z \rightarrow (w \neq z \land x < y < w) ~~~\land \\
		&z = x \rightarrow (x < w < y \vee y < w < x) ~~~\land \\
		&z = y \rightarrow (x < y < w \vee w < y < x)). 
\end{align*}
\end{small}

\begin{small}
\begin{align*}
	\Phi_{4, 6} = \forall x \exists y \forall z \exists w(& \notag \\
		&z = x \rightarrow (x < w < y \vee y < w < x) ~~~\land \\
		&z = y \rightarrow (x < y < w \vee w < y < x)). 
\end{align*}
\end{small}

\begin{small}
\begin{align*}
	\Phi_{4, 5} = \forall x \exists y \forall z \exists w(& \notag \\
		&z = x \rightarrow (x < w < y \vee y < w < x) ~~~\land \\
		&z = y \rightarrow (x < y < w \vee y < x)). 
\end{align*}
\end{small}

We have ``filled in the gaps'' for our base case of $r=4$ and we will now proceed by induction as we did in Theorem \ref{thm:main_lower}. A critical point is that all of these sentences for $\Phi_{4,k}$ for $4  < k \leq 9$ end with a universal and then an existential quantifier, meaning that Spoiler plays the next-to-last round on $L$ and the last round on $B$.

Suppose then that $g(r-1) \leq |L| < |B| < g(r)$ for $r > 4$. Spoiler adopts the same strategy as that described in the proof of Theorem \ref{thm:main_lower} for the $g(r)$ vs $g(r)-1$ game. One very minor nuance is that in the even $r$ case there is not always a central element for Spoiler to select for his 1st round move in $L$ since $|L|$ could be even. Similarly in the odd $r$ case there is not always a central element for Spoiler to select for his 1st round move in $B$ since $|B|$ could be even. It suffices, however, for Spoiler to play as close to the center as possible. As an example, the sentence $\Phi_{4, 9}$ corresponds to playing the left point among the two middle points in $L$ in the 9 versus 8 game. All other details of the argument are unchanged, including the fact that the last two rounds are played on $L$ and then $B$.
\end{proof}

Finally, we are able to prove the precise game theoretic analog of Theorem \ref{thm:g_for_game_play}.
\begin{thm} \label{thm:fund_thm_for_los}
Two linear orders, $B$ and $L$, with $|L| < |B|$, can be distinguished by a sentence with $r$ quantifiers iff $|L| <  g(r)$.
\end{thm}

\begin{proof}
If direction: If $|L| < |B| < g(r)$ then $B$ and $L$ can be distinguished by Lemma \ref{lemma:gap_closer}. If $|L| < g(r) \leq |B|$ then $B$ and $L$ can can be distinguished by the definition of $g$ (Definition \ref{def:g}). 

Only if direction: If $g(r) \leq |L| < |B|$ then $B$ and $L$ \textit{cannot} be distinguished by Corollary \ref{cor:main_upper}.
\end{proof}

\section{Conclusions} \label{sec:conc}
We have studied multi-structural (MS) games, which generalize \edashf games by being  played over  sets $\mathcal A$, $\mathcal B$ of structures rather than over a pair $A$, $B$  of individual structures.
Whereas \edashf games can capture exactly  the quantifier rank needed to describe a property, we showed that MS games can capture exactly the number of quantifiers needed to describe a property.
As a first application, we used them to determine the number of quantifiers needed to distinguish between linear orders of different sizes.  
This determination became quite nuanced as we uncovered several pitfalls of interest in their own right.

The quantifier count is a natural complexity measure but has received scant attention compared to the quantifier rank, number of distinct variable names, and measures of size and depth of the sentence body. 

We expect complexity differences to be magnified in studying other structures beyond linear orders, such as higher-dimensional lattices, rooted trees, and other classes of graphs.  As with the related ideas of Lotfallah \cite{Lotfallah04}, \ms games extend readily to second-order logic where they may bear on higher-order problems in descriptive complexity. For example, \ms games when adapted to second order logic can easily simulate Ajtai-Fagin games.

\section*{Acknowledgments}
We are indebted to Phokion Kolaitis for making Observation \ref{obs:phokion}. We thank Neil Immerman and Leonid Libkin for helpful discussions.

\bibliographystyle{alphaurl}
\bibliography{main}

\newpage
\appendix

\section{A Proof of Theorem \ref{thm:f}} \label{app:f}

\begin{proof} 
Note for future use that $f(1) = 1$ and  $f(r) = 2f(r-1) +1$ for $r>1$. 
The statement is clearly true for $r=1$.  Assume inductively that it is true for $r-1$. Let us refer to the bigger linear order as $B$ (for ``big'') and the smaller linear order by L (for ``little''). 

We first show that if the size of $L$  is less than $f(r)$, then Spoiler wins the $r$-round game. 
There are two possibilities, depending on whether the size of $B$ is odd or even.  Assume first that it is odd, say of size $2k+1$.  So the size of $L$ is at most $2k$, and since the size of $L$ is less than $f(r)$, it follows that $2k < f(r) = 2f(r-1) +1$, so $k \leq f(r-1)$.  In  the first round, Spoiler selects the middle point of $B$, call it $s_1$. There are then $k$ points to the left and $k$ to the right of $s_1$  in $B$. After Duplicator selects a point $d_{1}$  in $L$,  there will be either fewer than $k$ points to the left of $d_{1}$ in $L$ or fewer than $k$ points to the right of $d_{1}$ in $L$.   

In the former case, Spoiler now makes all of his moves to the left in either structure (forcing Duplicator to do the same in the other structure), and in the latter case (when there  less than $k$ points to the right of $d_{1}$ in $L$.), Spoiler makes all of his moves to the right  in either structure (forcing Duplicator to do the same in the other structure).  Spoiler now wins by the inductive assumption, since $k \leq f(r-1)$, and Spoiler has turned this into a game with less than $k$ elements in the smaller linear order. 

Assume now that the size of $B$ is even, say of size $2k$.  So the size of $L$ is at most $2k-1$, and since the size of $L$ is less than $f(r)$, we have
\[
2k-1 < f(r)  = 2f(r-1) + 1,
\]
so $2k - 2 < 2f(r-1)$, 
and so $k-1 < f(r-1)$.  

In  the first round, Spoiler selects the 
$k$th point of $B$, call it $s_1$. There are 
then $k-1$ points to the left and $k$ to the right of $s_1$ in $B$. Duplicator now selects a point $d_1$ in $L$.  If $d_1$ is within the first $k-1$ points in $L$,  then 
Spoiler makes all remaining moves to the left in both structures, since there are $k-1$ points to the left of $s_1$, fewer than $k-1$ points to the left of $d_1$, and $k-1 < f(r-1)$, so Spoiler wins by inductive hypothesis.  If $d_1$ is not within the first $k-1$ points in $L$,  then there are at most $k-1$ points to the right of $d_1$. But there are $k$  points to the right of $s_1$. Spoiler makes all remaining moves to the right  in either structures, and wins by inductive hypothesis since $k-1< f(r-1)$.

We now show that if the size of $L$  is at least  $f(r)$, then Duplicator wins the $r$-round game. 
If Spoiler selects $s_1$ within the first $f(r-1)+1$ points in $B$ or $L$,  then Duplicator selects $d_1$ in the other structure with $d_1$ = $s_1$.  There are now at least $f(r-1)$  points to the right of the first move in $L$,  and more points than that to the right of the first move in $B$. Since Duplicator can simply mimic Spoiler’s choices when Spoiler moves to the left of the first point chosen, Duplicator wins by the inductive hypothesis in considering moves on the right. 

Now assume that Spoiler selects $s_1$ beyond the first $f(r-1)+1$ points in $B$ or $L$.  There are two cases, depending on whether Spoiler moves in $B$ or $L$.  Assume first that Spoiler moves in $B$, and selects point $s_1$, which is the $k$th point from the right of $B$.  This splits into two subcases, depending on whether or not $k \leq f(r-1) + 1$. If $k \leq f(r-1) + 1$, then Duplicator selects $d_1$ in $L$ that is $k$ points from the right of $L$. There are then at least $f(r-1)$ points in $L$ to the left of $d_1$.  Duplicator simply mimics the moves of Spoiler on points to the right of $s_1$ or $d_1$, and uses her winning strategy for points to the left of $s_1$ and $d_1$, where there is a winning strategy by inductive assumption, since there at least $f(r-1)$ to the left of $d_1$ and more than that to the left of $s_1$.  

We now consider subcase 2, where $k > f(r-1) + 1$.   Let $I$ (respectively, $I'$) be the closed  interval in $B$ (respectively, $L$) 
where the left endpoint is at position $f(r-1)+1$ from the left of $B$ (respectively, $L$), and whose right endpoint is at position
$f(r-1)+1$ from the right of $B$ (respectively, $L$).  By assumption, the point $s_1$ in $B$ is inside of $I$. The interval $I'$ contains the point that is $f(r-1)+1$ from the left of $L$,  and so is nonempty. Assume that the point $s_1$ in $B$ is $m$ from the left side of $I$ and $n$ from the right side  of $I$.  Since $L$ is smaller in size than $B$, it follows that $I'$ is smaller in size than $I$. So there is a point $d_1$ in $I'$ that is at most $m$ from the left side  of $I'$, and at most n from the right side  of $I'$.  Since there are at least $f(r-1)$ points to the left of $d_1$ in $L$ and at least $f(r-1)$ points to the right of $d_1$ in $L$,  and since the number of points to the left (respectively, to the right) of $d_1$ in $L$ is at most the number of points to the left (respectively, to the right) of $s_1$ in $B$, Duplicator can win by making use of her winning strategy on moves to the left of $s_1$ or $d_1$, and make use of her winning strategy on moves to the right of $s_1$ and $d_1$.

Assume now that Spoiler selects point $s_1$  in $L$.  By assumption, $s_1$  is beyond the first $f(r-1)+1$ points in $L$.    Assume $s_1$ is the $k$th point from the right of $L$.  Then Duplicator selects $d_1$ as the $k$th point from the right of $B$. On moves to the right of $d_1$ or $s_1$, Duplicator simply mimics Spoiler’s moves, an on moves to the left of $d_1$ or $s_1$, Duplicator has a winning strategy by inductive assumption, since there are more than $f(r-1)$ points to the left of $s_1$ in $L$,  and more than that to the left of $d_1$ in $B$.  So again, Duplicator wins.  This concludes the proof of the theorem.
\end{proof}

\section{Game based proof of Lemma~\ref{lemma:g_underline_of_4}}\label{sec:app-g_underline_of_4}
\begin{proof} [Proof of Lemma~\ref{lemma:g_underline_of_4} via Games]
We first consider the singleton $10$ \vs $9$ base case.  We later consider the case where $|L| < 9$, followed by $|B| > 10$, and lastly the case of many linear orders $\mathcal B = \{B_i\}, \mathcal L = \{L_j\}$, of sizes $|B_i| \geq 10, |L_j| \leq 9$.  In the $10$ \vs $9$ game, Spoiler will play first by playing the central element, $L(5)$, on the $L$ side. Without loss of generality, Duplicator then plays all possible moves on the $B$ side. See Figure \ref{fig:10_vs_9_counter_move1} for an illustration. 

\begin{figure} [ht]
\centerline{\scalebox{0.35}{\includegraphics{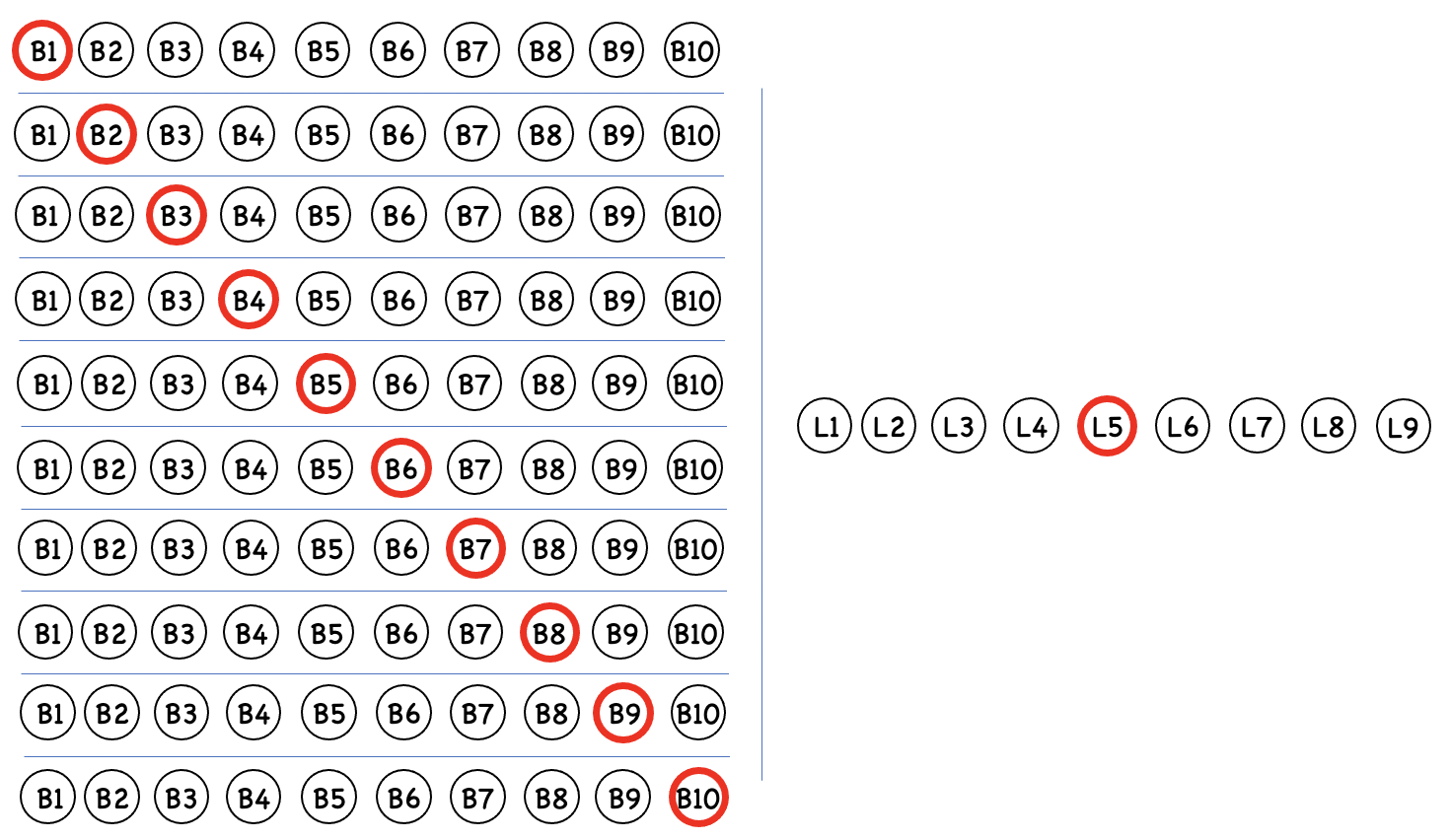}}}
\caption{First Spoiler plays $L(5)$ on the $L$ side. Duplicator responds by playing all possible moves on the $B$ side. When discussing a given $B$ board, given that there are $10$ elements, there will either be more elements to the right or more elements to the left of the 1st move played. We refer to the side that has more elements as the ``long side'' and the side with fewer elements as the ``short side.''  Thus, in the figure, on the board in which Duplicator has played $B(5)$, the long side is to the right and the short side is to the left.}
\label{fig:10_vs_9_counter_move1}
\end{figure}

Since there are $10$ elements on each board on the $B$ side, there will necessarily be more elements on one side or the other of any move on any particular board. Refer to the side that has more elements as the ``long side'' and the side with fewer elements as the ``short side.'' Spoiler will now make his 2nd round moves on $B$, playing as close as he can to the middle of the long side of every board. Since there are at least $5$ elements on the long side of every board, there will necessarily be at least two unplayed elements to either side of Spoiler's 2nd round move on the long side of each board. See Figure \ref{fig:10_vs_9_counter_move2_spoiler}.  
\begin{figure} [ht]
\centerline{\scalebox{0.35}{\includegraphics{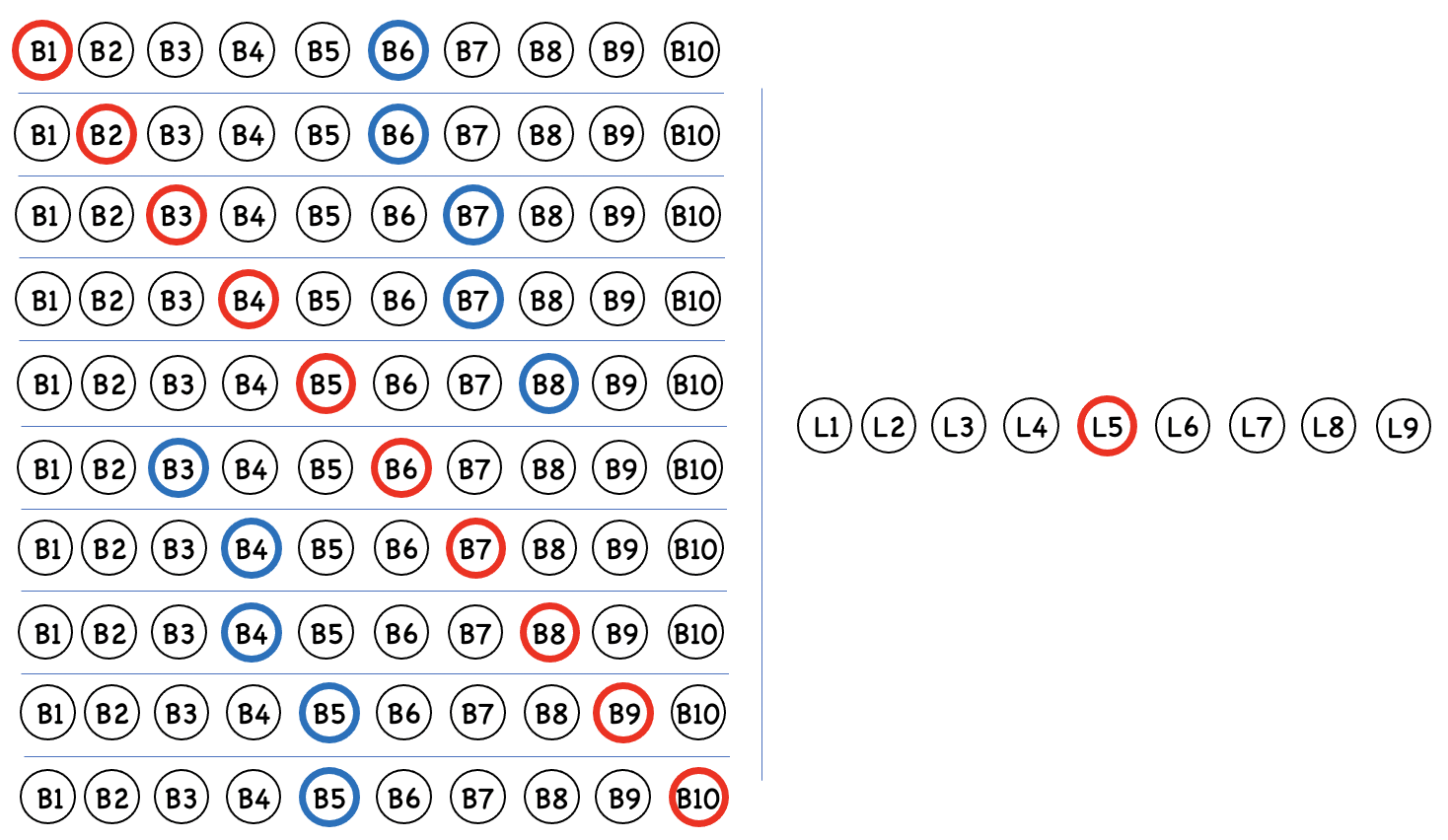}}}
\caption{Spoiler's 2nd moves on the $B$ side in blue), associated with each possible 1st round move on this side by Duplicator. The moves are all in the ``middle'' of the ``long sides.'' See the text for a description of these terms.}
\label{fig:10_vs_9_counter_move2_spoiler}
\end{figure}

Refer to these moves that have two unplayed elements to either side of them as ``middle moves.'' Duplicator then follows up, playing every possible move on the $L$ side. 
The possible moves are depicted in blue in Figure \ref{fig:10_vs_9_counter_move2}. There is plainly no value to playing on top of the 1st round move here so we omit that option.
\begin{figure} [ht]
\centerline{\scalebox{0.40}{\includegraphics{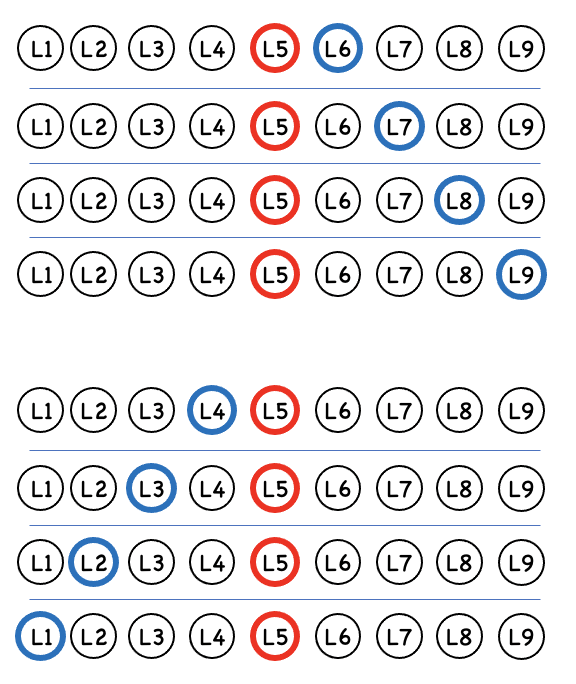}}}
\caption{Possible 2nd moves by Duplicator (in blue) in response to Spoiler playing middle moves from the long side of all $B$ boards. The play-on-top move is omitted, as discussed in the text.}
\label{fig:10_vs_9_counter_move2}
\end{figure}

For his 3rd move, Spoiler will now play on $L$, replying to each possible play of Duplicator using the green moves in Figure \ref{fig:10_vs_9_counter_move3}.
\begin{figure} [ht]
\centerline{\scalebox{0.40}{\includegraphics{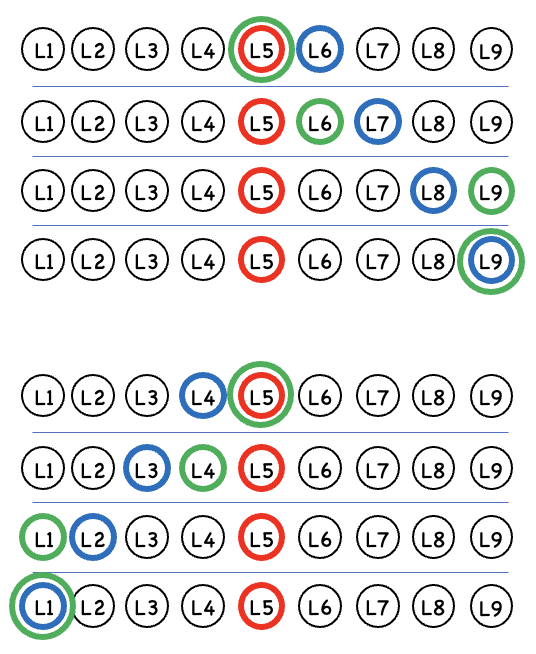}}}
\caption{3rd moves of Spoiler (in green) in response to each possible 2nd move of Duplicator (in blue).}
\label{fig:10_vs_9_counter_move3}
\end{figure}
In order to keep an isomorphism with the elements in boards 1 or 5 of Figure \ref{fig:10_vs_9_counter_move3}, Duplicator will have to play on top of her 1st move on $B$. Note that this move breaks the isomorphism with all other $L$ boards except board 1 if the $B$ board had its long side on the right of the first move, and except board 5 if the $B$ board had its long side on the left of the first move. Spoiler is going to play his 4th moves from the $B$ boards. In case the second move on $B$ was right of the first move (in other words the long side was to the right), Spoiler will play to the right between the 1st and 2nd moves. Since there is no corresponding move on the 1st board, this move breaks the isomorphism with all $L$ boards. Analogously, if the 2nd move on $B$ was to the left of the 1st move, Spoiler will play to the left between the 1st and 2nd moves, again breaking the isomorphism with all boards. Since the arguments continue to be completely parallel, with plays on the top four boards of Figure \ref{fig:10_vs_9_counter_move3} corresponding to cases where the long side of $B$ was to the right, and play on boards 5 through 8 of Figure \ref{fig:10_vs_9_counter_move3} corresponding to cases where the long side of $B$ was to the left, we shall focus just on the top four boards. We have ruled out the case where Duplicator makes a 3rd round play that is on top of her 1st round play. To try to keep an isomorphism going with board 4 in the Figure, Duplicator must play on top of her 2nd move. But in this case Spoiler plays to the right of his 2nd move and there is no analogous move on board 4. Hence Duplicator cannot keep an isomorphism going with board 4. Next, to keep an isomorphism going with board 2 in Figure \ref{fig:10_vs_9_counter_move3}, Duplicator must play to the right between her 1st and 2nd moves. But then Spoiler plays a second time to the right between his 1st and 2nd moves, in the at least one additional unplayed element and Duplicator cannot reply in kind on the 2nd board. An attempt to maintain the isomorphism with board 3 in the Figure is seen to be fruitless with an analogous argument. Duplicator must play to the right of the 2nd move on $B$ boards that have their long side to the right. Spoiler then picks a second unplayed element to the right of the 2nd move played on these boards and there is no analogous move on the 3rd board on the $L$ side. It is thus impossible to maintain an isomorphism with any of the boards in Figure \ref{fig:10_vs_9_counter_move3} and hence Spoiler wins this 10 versus 9 game. 

Suppose now that $|L| < 9$. On his first move Spoiler will play a move that is as close to the center of $L$ as possible. Duplicator will then just have fewer possible moves in the 2nd round when she plays on $L$ because $|L|$ is smaller than in the prior analysis. Spoiler still exploits the fact that there are two unplayed elements both to the left and right of the 2nd round moves on the same side of the 1st round moves on $B$ but not on $L$. If, for a 2nd round move, Duplicator plays immediately to the right or immediately to the left of Spoiler's round $1$ move, Spoiler will, on his 3rd round move, play on $L$ on top of his 1st move (just as before). If Duplicator plays an end move, which is not simultaneously an immediate neighbor of Spoiler's 1st move, then Spoiler will play on top of that move (again, as before). On the other hand, if Duplicator plays a move that is two to the right (alternatively, two to the left) of the 1st round move, and the move is not also an end move, then Spoiler will play immediately to the right (alternatively, immediately to the left) of center. Analogously, a Duplicator move that is a 2nd-from-end element and not covered in any other case, will be responded to with an end move on the same side. In all cases the analysis is exactly as before and results in a Spoiler victory. Finally, if $|B| > 10$, the analysis is not particularly different than what we've seen; there is just a bit more ``room'' when picking Spoiler's 2nd round moves, which again must leave at least two unplayed elements both to the right and left. All other aspects of the analysis are unchanged. Finally, if we have a game with multiple linear orders, $\mathcal B = \{B_i\}, \mathcal L = \{L_j\}$ on the respective sides, with $|B_i| \geq 10$ and $|L_j| \leq 9$, then Spoiler first plays on each liner order in $\mathcal L$ as described above, and then alternates -- the fact that there happen to be multiple linear orders of different sizes has no impact on the strategy. The lemma is therefore established.
\end{proof}

It is worth noting that Spoiler's ability to play on top of existing moves was essential to the above Spoiler-winning strategy. Suppose such a move were prohibited. Consider the situation after Duplicator's 2nd round moves -- see Figure \ref{fig:10_vs_9_counter_move2}. Consider just the top $4$ boards in this figure and suppose for the moment that our 3rd round moves were constrained to be on $L(6)$--$L(9)$. By the analysis showing that the $5$ \vs $4$ game is Duplicator-winning (proof of Lemma \ref{lemma:g_upper_of_3}), recall that playing two 3rd round moves with both moves either to the left of the 2nd round (blue) moves or to the right of the 2nd round moves would lead to a Duplicator victory. Thus, under the assumption that Spoiler does not play on top of an existing move, he would have to play moves on two boards in the $L(1)$--$L(4)$ range. But then, for a 3rd round move, Duplicator could mimic any one of these Spoiler moves on the $4$th $B$ board (e.g., playing $B(i)$ if Spoiler played $L(i)$). It is then evident that any 4th round move on this 4th $B$ board would be fruitless, while if  Spoiler makes his 4th round moves on $L$, one of the moves on the two boards in which Spoiler previously played on $L(1)$--$L(4)$ will allow for a Duplicator victory with respect to the same $4$th $B$ board. 

It is worth calling out this last fact explicitly: 

\begin{observation} \label{obs:play-on-top} If Spoiler were not able to play on top of existing moves, the M-S game on linear orders of sizes $10$ vs.\ linear orders of size $9$ would be winnable by Duplicator.
\end{observation}

\end{document}

%% file: macros.tex
\newcommand{\old}[1]{{}}

\newcommand{\bi}{\begin{itemize}}
\newcommand{\ei}{\end  {itemize}}
\newcommand{\bt}{\begin{tabbing}}
\newcommand{\et}{\end  {tabbing}}
\newcommand{\be}{\begin{enumerate}}
\newcommand{\ee}{\end  {enumerate}}

\newcommand{\fraisse}{Fra\"{i}ss\'{e} }

\usepackage{ulem}
\usepackage{xcolor}

\makeatletter
\def\begin@lgo{\begin{minipage}{1in}\begin{tabbing}
        \quad\=\qquad\=\qquad\=\qquad\=\qquad\=\qquad\=\qquad\=\kill}
\def\end@lgo{\end{tabbing}\end{minipage}}

\makeatother

\theoremstyle{definition}

\newtheorem{observation}[thm]{Observation}
\newtheorem{observ}[thm]{Observation}

\newcommand{\set}[1]{\{#1\}}
\newcommand{\qr}{\mathrm{qr}}
\renewcommand{\gets}{\!\leftarrow\!}

\newcommand{\AND}{\; \wedge \;}

\newcommand{\ef}{Ehrenfeucht-\fraisse}
\newcommand{\edashf}{E-F }
\newcommand{\vs}{vs. }
\newcommand{\ms}{MS }
\newcommand{\MS}{MS }